\documentclass[12pt, draftcls, onecolumn]{IEEEtran}
\usepackage[utf8]{inputenc}
\usepackage{epsfig}

\usepackage{amsmath}
\usepackage{amsthm}
\usepackage{amssymb}
\usepackage{array}
\usepackage{float} 
\usepackage{ctable}
\usepackage{multirow}
\usepackage{footnote}
\usepackage{xfrac}
\usepackage{graphicx}
\usepackage{tikz}
\usetikzlibrary{arrows,positioning} 

\tikzset{
    >=stealth',
    pil/.style={
           <-,
           thin,
           shorten <=1pt,
           shorten >=1pt,}
}

\usepackage{subcaption}
\usepackage{mathrsfs} 

\usepackage{amsfonts}
\usepackage{cite}
\usepackage{algorithmic}
\usepackage{url}

\makeatletter
\def\ifundefined{\@ifundefined}
\makeatother

\def\th{^{th}}
\def\pd{\frac{\partial}{\partial q_i}}
\newcommand{\mcf}[1]{p\left( \norm{\bs{#1}}_1 \right)}

\def\pdq{\frac{\partial}{\partial q_i}}
\newcommand{\diag}{\mathop{\mathrm{diag}}}
\newcommand{\bs}[1]{\boldsymbol{#1}}
\newcommand{\norm}[1]{\big\lVert#1\big\rVert}

\newcommand{\normb}[1]{\big\lVert \bs{#1} \big\rVert_1}

\usepackage{xfrac}
\newcommand{\slfrac}[2]{\left.#1\middle/#2\right.}

\newtheorem{theorem}{Theorem}

\newtheorem{proposition}{Proposition}

\newtheorem{definition}{Definition}
\newtheorem{remark}{Remark}
\theoremstyle{definition}
\newtheorem{assumptionletter}{Assumption}

\theoremstyle{definition}

\usepackage{setspace}

\IEEEoverridecommandlockouts

\theoremstyle{definition}
\newtheorem{condition}{Condition}

\usepackage{autonum}

\usepackage{color}
\newcommand{\modification}[1]{{\leavevmode\color{black} #1 }}

\begin{document}

\title{ 
Incentives-Based Mechanism for Efficient Demand Response Programs 
}

\author{Carlos Barreto, Eduardo Mojica-Nava, and Nicanor Quijano 
 \thanks{
 This work has been supported in part by Proyecto CIFI 2012, Facultad de Ingeniería, Universidad de los Andes, also by project Silice III Codensa-Colciencias under Grant P12.245422.006, and also by the program J\'ovenes Investigadores 2012, by Colciencias
} 
 \thanks{  Carlos Barreto is with the 
 Computer Science Department, University of Texas at Dallas
(e-mail: carlos.barretosuarez@utdallas.edu) }
\thanks{  Nicanor Quijano is with the 
 Electrical and Electronics Department, Universidad de los Andes. Bogotá, Colombia
(e-mail: nquijano@uniandes.edu.co) }
\thanks{Eduardo Mojica-Nava is with the 
 Electrical and Electronics Department, Universidad Nacional de Colombia. Bogotá, Colombia
(e-mail: eamojican@unal.edu.co) 
}}

\date{\today}

\maketitle

\begin{abstract}

In this work we 
investigate
the inefficiency 
of the electricity system with strategic agents.
\modification{
Specifically, we prove that 
without a proper control 
the total demand of an inefficient system is at most twice the total demand 
of the optimal outcome.
We propose 
an incentives scheme that promotes optimal outcomes in the inefficient electricity market. 
The economic incentives can be seen as an indirect revelation mechanism that allocates resources using a one-dimensional message space per resource to be allocated. The mechanism does not request private information from users and is valid for any concave customer's valuation function.  
}
 We propose a distributed implementation of the mechanism using population games 
 and evaluate the performance of four popular dynamics methods in terms of the cost to implement the mechanism. 
 We find that the achievement of efficiency in strategic environments might be achieved at a cost, which is dependent on both the users' preferences and the  dynamic evolution of the system. 
 Some simulation results illustrate the ideas presented throughout the paper.

\end{abstract}



\section{Introduction}

The \emph{smart grid} (SG) concept has entailed profound changes in the conception of electricity systems. Specifically, efficiency in both electricity generation and consumption is one of the main goals of the SG \cite{santacama2010}. 
\modification{
One of the main characteristics of the electricity system is that 
the technological development might be insufficient by itself to achieve the desired goals \cite{Honebein2011}. 
Hence, 
}
efficiency is expected to be achieved by means of an active cooperation of consumers in the electricity systems.  \emph{Demand response} (DR) programs arise as a tool intended to promote cooperation of consumers with the electricity system. DR programs deal with the problem of providing economic incentives to consumers in order to modify their electricity consumption behavior. 
Incentives proposed in the literature use price mechanisms, by which either
higher or lower consumption is encouraged through changes in the 
electricity prices.
However, the design of incentives is not a trivial task, because the price scheme might impact the 
robustness of the system \cite{dahleh2010}.
Also, incentives must be designed having into account the characteristics of the consumers, who might be either \emph{price takers} or \emph{price anticipators}.
On the one hand, a price taker  makes decisions without considering 
future implications of its decisions. On the other hand, a price anticipator
(or strategic agent in the context of \emph{game theory}) makes decisions taking into account the consequences of its actions on the system. 
\modification{
Particularly, 
the behavior of price anticipators 
might 
lead to a degradation of the system's efficiency,   
which is known as the \emph{price of anarchy} \cite{papadimitriou2001algorithms}, 
or the \emph{tragedy of the commons}, 
when the efficiency loss is accompanied by an overuse of resources \cite{hardin1968tragedy}.  
In this work, we show that the electricity system experiences the  {tragedy of the commons}, and particularly, we show that the total demand in an inefficient system is at most twice the total demand in an efficient system (when we consider average cost pricing).
This result complements previous works that investigate the efficiency loss of resource allocation processes \cite{tardos, johari2006scalable, johari2004efficiency, johari2005efficiency, Johari09}.
}

\modification{
The inefficiency of resource allocation processes has inspired the search of schemes that achieve optimal outcomes. 
Efficiency 
}
in a society of strategic agents might be accomplished with \emph{mechanism design,} which addresses the problem of designing the rules of a game to reach the desired outcome \cite{jackson2000mechanism, hurwicz06}. An outstanding result in mechanism design is the Vickrey-Clarke-Groves mechanism (VCG), which is a 
\emph{direct revelation mechanism} that guarantees efficient outcomes in dominant strategies  \cite{Vic61, clarke71, groves73}. 
\modification{
The VCG mechanism requests private information from users (such as preferences or payoff functions) and assigns a central agent
the tasks of 
gather information and compute the optimal resource allocation. 
}
In particular, the VCG mechanism provides incentives to elicit private information, even if users are unwilling to report their true preferences due to strategic 
or privacy issues. 
\modification{
However, the implementation of the VCG mechanism might not be practical. For instance,  
the revelation of private information might affect correlated activities of users  \cite{Rothkopf07} and also might need 
 high amounts of both communication and
computation resources to implement the optimal solutions \cite{hurwicz06}.
}

A mechanism
that allocates resources without private information 
is the Kelly mechanism \cite{kelly1997charging, kelly1998rate}. In this mechanism users send a bid to a central planner
who
assigns resources following a proportional fairness criteria. 
However, there is an efficiency loss if users are strategic 
\cite{johari2004efficiency}, which inspired the design of efficient mechanisms in strategic environments. In \cite{yang2006vcg}, an efficient one-dimensional mechanism based on the Kelly and VCG mechanisms is proposed. 
\modification{
On the other hand, \cite{Johari09} presents
a more general mechanism for convex environments 
in which each user reports her payoff function. 
}
These mechanisms use a one-dimensional message space because the payoff functions are parametrized by a single parameter.


%
%

\modification{ In this work,
we propose a scheme of economic incentives that
can be seen as an \emph{indirect revelation mechanism} based on the Clarke 
pivot mechanism \cite{AlgorithmicG}. 
The main feature of our mechanism is that, \modification{unlike the VCG mechanism,}  it does not require private information from users. 
Furthermore, 
it 
coordinates the 
resource allocation process
to achieve optimal outcomes  
utilizing a one-dimensional message space per each resource to be allocated.
}
Specifically, the mechanism entrusts the computational tasks among users, who 
maximize their own 
\modification{surplus}
based on the total demand (which in turn is calculated and broadcasted by a central agent). Thus, users avoid revelation of private information (e.g., preferences), but are required to report the aggregate consumption of their appliances during some time periods. 
It is noteworthy that the mechanism implements efficient outcomes regardless of the form of the valuation function (as long as it is concave). 
\modification{
Nevertheless, this incentives scheme might require external subsidies to be implemented, which means that the benefit obtained in the optimal outcome is not enough to fund the incentives scheme. 
}

\modification{
The optimization problem of each user is seen as a local resource allocation problem, in which the electric energy is the resource that should be allocated in different time periods.
We extend the work presented previously in \cite{barreto2013design} by implementing the mechanism using evolutionary dynamics to distribute efficiently the power consumption of each user.
Evolutionary dynamics are built based on myopic learning, which resembles subject's behavior in repeated games  \cite{chen04learning, Healy06learning}. With a myopic dynamic it is not necessary to broadcast real time information about the aggregated demand, because customers update their strategy based on the previous actions.
}

\subsection{Contributions}

The main contributions of this paper can be summarized as follows: 
1) we formulate the demand response problem as a tragedy of the commons dilemma, highlighting that the efficiency loss 
\modification{
takes place with an overuse of resources of
at most twice the total demand in an efficient system.
This result can be utilized to prove that the peak to average ratio can be reduced to a half in an efficient system;
}
2) we propose a novel scheme of economic incentives for achieving optimal demand profiles in a population of strategic agents. The mechanism uses a one-dimensional message space per resource to be allocated, is valid for any concave  electricity valuation function, and satisfies the properties of \modification{budget deficit} and individual rationality; and
3) we present a distributed implementation of the mechanism using population games,
and we evaluate the performance of four popular evolutionary dynamics, namely logit dynamics, replicator dynamics, Brown-von Neumann-Nash dynamics, and Smith dynamics in terms of the cost to implement the mechanism. 
We highlight that the achievement of efficiency in strategic environments might be achieved at a cost, which is dependent on both the users' preferences and the  dynamic evolution of the system.

\subsection{Related Literature}

\modification{
Let us elaborate on some properties of the mechanism and related literature.
%
%
The structure of the economic incentives is similar in philosophy to the mechanism in \cite{Groves&Ledyard1977}. However, \cite{Groves&Ledyard1977} presents a direct revelation mechanism in which users must communicate their marginal payoff functions.
Hence, the mechanism might need a high-dimensional message space to describe the required function.
%
Previous works analyze the dependence of the message space with the efficiency of the mechanism. 
For instance, \cite{reichelstein1988game} analyzes the minimum space required to implement optimal allocations in exchange economies, which is larger than the number of participants and resources. 
Furthermore, \cite{healy2012designing} shows that one-dimensional message spaces are not suitable to achieve optimal allocations on contractive games (games whose best response is a contraction mapping).
In this work we manage to design a mechanism that uses a one-dimensional space (to allocate a single resource) due to the properties of the system. 
Specifically, the wealth of an agent depends on her own action and the aggregate actions of other agents. Thus, we can decentralize decisions by sending to each agent the aggregated demand, which is one-dimensional
(population games satisfy the conditions to reduce the message space \cite{sandholm_book}).
Also, we guarantee stability through the properties of potential games, and consequently, we do not impose the restrictions of \cite{healy2012designing} in our work. 
}

\modification{
Furthermore, works as in \cite{teneketzis12a, teneketzis12b} propose decentralized mechanisms to allocate efficiently bandwidth in networks. These mechanisms satisfy the budget balance property, but require a multidimensional message space, because users must report their desired allocation and the price they are willing to pay for each link of the network.
Our mechanism might not guarantee the  budget balance property, but this shortfall is compensated by the reduction of the message space. 
%
}

\modification{
Our work is closely related to other indirect revelation mechanisms that implement efficient outcomes with limited information using price incentives \cite{mohsenian, sandholm2005negative}. 
However, our work differs from \cite{mohsenian} in 
that they 
minimize the energy cost, while we maximize the aggregate surplus of the society. 
Also, in \cite{mohsenian} each user has to broadcast the scheduled daily energy consumption to all other users, while we assume that a central agent broadcast this information to all users. 
On the other hand, \cite{sandholm2005negative}
models the society's behavior using evolutionary dynamics, which converge to the optimal solution because the system is modeled as a potential game. 
A limitation of \cite{sandholm2005negative} is that the payoff functions must be additively separable in the private information of users.
In our work we 
 model the strategy of a customer as the mixed strategy of a population. Thus, an agent is not limited to a fixed set of strategies (as in \cite{sandholm2005negative}), but can choose from a large but finite set.

}



This work is inspired in previous applications of mechanism design in efficient resource allocation problems, notably \cite{johari2006scalable}. 
\modification{
}
In particular, we consider the DR problem as a resource allocation problem of multiple goods. 
We depart from other resource allocation approaches such as \cite{marden14bounds}, in which we do not specify the payoff functions that lead to a desired outcome. Instead, we design the rules of the game to achieve the desired result without modifying the valuation function of each agent.

\subsection{Organization of the Paper}

The rest of the paper is outlined as follows. 
In Section~\ref{sec:backgnd} we introduce the model of the electricity system and prove that its inefficiency is caused by selfish agents that have incentives to over-exploit resources. In Section~\ref{sec:incentives}, we design  an indirect revelation mechanism that uses an incentive scheme to  achieve the optimal outcomes in strategic frameworks.
In Section~\ref{sec:implementation}, we introduce the distributed implementation of the mechanism. 
Section~\ref{sec:sim} presents some examples to illustrate the ideas developed throughout the paper. Finally, Section~\ref{sec:concl} presents some conclusions and future directions.


\section{Problem Statement}\label{sec:backgnd}

In this section we introduce the market model of an electricity system and present two solution concepts arising in both ideal and strategic environments. Then, we prove the inefficiency of the electricity system in strategic environments and discuss its consequences in electricity markets.

\subsection{Electricity System Model}\label{sec:model}

We consider an electricity market composed by three parties, namely the generator, the customers, and the independent system operator (ISO). In this case, a single generator provides energy to several customers, while the ISO maintains a balance between supply and demand (market clearing). We define the society as a set with $N$ customers denoted by $\mathcal{P} = \{1,\ldots,N\}$.

\modification{
We take into account users that  might have time varying consumption preferences. These preferences are associated with their changing necessities of energy along 
the day. 
This feature is modeled by partitioning a day in $T$ time intervals in which users have roughly the same consumption preferences. 
In order to make the problem tractable, we assume that the system's parameters are independent in each time period and 
decompose the demand response problem into $T$ problems.
Henceforth we deal with a single subproblem and in
Section~\ref{sec:implementation} we couple multiple time periods using the definition of a population game.
}


We denote the daily consumption of the $i\th$ customer with a real number $q_i\in \mathbb{R}_{\geq 0}$. 
Without loss of generality, we restrict the electricity consumption to
zero or positive values, i.e., $q_i\geq 0$, for all agents $i\in\mathcal{P}$.
Likewise, the electricity consumption of the society is denoted by the 
vector $\bs{q} = [q_1, \ldots, q_N]^\top \in \mathbb{R}_{\geq 0}^{N}$ and the 
vector $\bs{q}_{-i} = [q_1, \ldots, q_{i-1}, q_{i+1}, \ldots,  q_N]^\top \in \mathbb{R}_{\geq 0}^{N - 1}$ represents the consumption of the society, except for the $i\th$ customer.

Ideally, each customer regulates its electricity consumption in function of its own preferences and the electricity price. 
In particular, rational agents might try to maximize their profit, which is defined as the benefit earned minus the cost associated with the energy consumed.
The benefit might be represented by means of a  \emph{valuation function} $v_i:\mathbb{R}_{\geq0}\rightarrow \mathbb{R}$, where $v_i(q_i)$ represents the economic value that the $i\th$ user assigns to $q_i$ electricity units. 

In a market economy the consumption is also limited by the cost associated to it. Hence, besides their valuation of energy, customers also  take into account the electricity price to make consumption decisions.
In particular, the cost of the energy  is denoted by $\lambda$, and consequently,  the profit of the $i\th$ consumer can be represented by 
\begin{equation}\label{eq:utility}
U_i(\bs{q}) = v_i(q_i) -  q_i \lambda.
\end{equation}

We assume that the generation is coordinated 
by the ISO in order to guarantee the balance between generation and the total demand, i.e., to guarantee that $g = \sum_{i\in\mathcal{P}} q_i$. 
\modification{
Furthermore, we assume that the ISO regulates the market by enforcing an \emph{average cost price} rule, which imposes some limits on the price charged to customers. 
This price scheme charges customers an amount equal to the cost necessary to create the product and it is used to avoid market manipulation in natural monopolies, such as the electricity system \cite{brown1983marginal}. Average cost pricing is often used because it is easier to measure compared to the marginal cost, which requires the whole cost curve. 
%
The average cost price function is defined as 
\begin{equation}\label{eq:avg_price}
p(g) = \slfrac{C(g)}{g},
\end{equation}
where $g $ represents the total demand and $C(g)$ is the production cost associated with a generation of $g$ energy units. 
}
%
%
\modification{
}
%
Using the  average cost price scheme, we can express the profit of the $i\th$ agent in Eq. (\ref{eq:utility}) as
\begin{equation}\label{eq:u_i_}
U_i(\bs{q}) = U_i(q_i, \bs{q}_{-i}) = v_i(q_i) -  q_i p\left( \norm{\bs{q}}_1 \right),
\end{equation}
where 
$||\cdot||_1$ is the 1-norm. 


\subsection{Market Equilibrium}

A market equilibrium is reached when the reiterated interaction between buyers and sellers is in balance. Here, we introduce the equilibrium conditions of both ideal economies and economies with strategic agents, which are used to show the inefficiency of the electricity system in strategic environments.

\subsubsection{Optimal Equilibrium}
In an ideal market customers behave as price takers and consequently,  and their behavior lead to an optimal outcome in the sense of Pareto (i.e., an outcome in which no individual can made better without making someone else worse off), which 
maximizes the 
aggregate benefit of producers and consumers. 
Accordingly, the demand profile $\bs{\mu} \in \mathbb{R}_{\geq0}^{N}$ that maximizes the aggregate surplus is a solution to the following optimization problem:
\begin{equation}
\begin{aligned}
& \underset{\bs{q}}{\text{maximize}}
& &  \sum\nolimits_{i\in\mathcal{P}} U_i(\bs{q}) \\
& \text{subject to}
& & q_i \geq 0, \, i =\{1,\ldots,N\}. 
\end{aligned}
\label{eq:opt_problem}
\end{equation}
\modification{
Note that the profit of the producer is zero under an average cost price scheme, therefore, the aggregate surplus only includes the surplus of consumers.
}

\modification{
The following assumptions state that the more an individual consumes, the larger is her satisfaction. 
However, each additional unit of satisfaction (or marginal utility) decreases with each additional unit of consumption. 
This is known as the diminishing marginal utility property.
%
%
}

\modification{
\begin{assumptionletter} \label{as:1} \
\textit{
\begin{enumerate}
 \item The valuation function $v_i(\cdot)$ is differentiable, concave,  non-decreasing, and satisfies $v_i(0) = 0$.
 \item The generation cost function can be expressed as $C(z) = z p(z)$, where $p(\cdot)$ is differentiable, convex, and non-decreasing.
\end{enumerate}
}     
\end{assumptionletter}
}
\modification{
With these assumptions the customer's surplus is a concave function. 
This is a necessary condition by the second theorem of welfare to guarantee that the any Pareto-efficient allocation can be supported as a competitive equilibrium \cite{mas1995microeconomic}.
The existence of a unique market equilibrium $\bs{\mu}$ inside the feasible area (i.e., $\bs{\mu}$ belongs to $\mathbb{R}_{\geq 0}^{N}$) is ensured  if the following assumption is satisfied.

\begin{assumptionletter} \label{as:2} 
         \textit{
Let us consider an arbitrary demand profile $\tilde{\bs{q}}$ such that $\tilde{q}_i = 0$ and $\tilde{q}_j \geq 0$ for any $i,j\in\mathcal{P}$ such that $i\neq j$. In this case, the $i\th$ user has incentives to increase its demand. In other words, the marginal valuation at $\tilde{q}_i = 0$ is greater than the unitary electricity price, i.e., 
 \begin{equation} 
 \pdq  U_i ( \bs{q}) \bigg|_{\bs{q}=\tilde{\bs{q}}} = 
 \pdq v_i(0) - p\left( 
 \norm{ \tilde{\bs{q}} }_1 
 \right)   \geq 0,
 \end{equation}
 for all agent $i\in \mathcal{P}$.
 }
\end{assumptionletter}
}
\begin{remark}
These assumptions are reasonable in the context of economic theory \cite{mas1995microeconomic}. Notice that Assumptions \ref{as:1} and \ref{as:2} imply that the problem in Eq. (\ref{eq:opt_problem}) has an optimal $\bs{\mu}$ that satisfies the following first order conditions (FOC):
\begin{multline}
\pd \sum\nolimits_{i\in\mathcal{P}} U_i (\bs{q}) \bigg|_{\bs{q}=\bs{\mu}} = 
\pd v_i(q_i)   
\\
-\mcf{q} -
\left( 
\norm{ \bs{q} }_1 
\right) \pd  \mcf{q} 
\bigg|_{\bs{q}=\bs{\mu}}= 0.
\label{eq:FOC1}
\end{multline}
\end{remark}


\subsubsection{Nash Equilibrium} \label{sec:game}

\modification{
We model the strategic interactions by means of a \emph{Cournot game}, in which users select the quantity they want to consume \cite{cournot1838recherches, mas1995microeconomic}.
%
In the electricity system
there is a conflict between agents, because their actions
impose externalities on the society through the price signals. }
If the society is finite, then the consumption of each agent might have a significant impact on the electricity prices and might affect the profit of other agents. 
%
%
Consequently, the electricity system might be seen as a game in which each agent is selfish and endeavors to maximize 
independently its own profit. 

The game can be defined as the 3-tuple  $G = \langle \mathcal{P}, (S_i)_{i \in \mathcal{P}}, (U_i)_{i \in \mathcal{P}} \rangle$, where $\mathcal{P}$ is the set of players (or customers), $S_i = \mathbb{R}_{\geq 0}$ is the set of available strategies \modification{(consumption of electricity)} of each player, and $U_i: S_1 \times \dots \times S_N \to \mathbb{R}$ is the surplus function of the $i\th$ player as a function of its own actions as well as the  strategies of other players.
The equilibrium concept used in game theory is   the Nash equilibrium \cite{nash1951non, fudenberg-tirole}. In particular, the Nash  equilibrium of the game $G$, denoted by $\bs{\xi} \in\mathbb{R}_{\geq 0}^{N}$, satisfies
\begin{equation}
 U_i (\xi_i,\bs{\xi}_{-i}) \geq U_i (q_i,\bs{\xi}_{-i})  , \, \text{for all } q_i\in\mathbb{R}_{\geq 0}, 
\end{equation}
for all agents $i\in\mathcal{P}$.
 The Nash equilibrium $\bs{\xi}$ is a solution to
 the following maximization problem that every agent $i\in\mathcal{P}$ attempts to solve independently
\begin{equation}\label{eq:opt_problem_games}
\begin{aligned}
& \underset{q_i}{\text{maximize}}
& & U_i (q_i,\bs{q}_{-i}) =  v_i(q_i) - q_i p\left( \norm{ \bs{q} }_1 \right)\\
& \text{subject to}
& & q_i \geq 0,  i =\{1,\ldots,N\}.
\end{aligned}
\end{equation}

Assumptions \ref{as:1} and \ref{as:2} imply that the Nash equilibrium $\bs{\xi}$ of the game $G$  satisfies the following FOC:
\begin{multline}
\pd U_i (q_i, \bs{q}_{-i}) \bigg|_{\bs{q} = \bs{\xi}} = \pd v_i(q_i) \\
 -  \mcf{q} -  q_i \pd  \mcf{q} \bigg|_{\bs{q} = \bs{\xi}} = 0,
 \label{eq:FOC2}
\end{multline}
%
%
\modification{
Furthermore, we can use the results by Rosen \cite{rosen1965} to prove that the Nash equilibrium is unique provided that the strategy space is convex.

}

\modification{
\begin{remark}
 The problems described in this section have solution as long as the strategy space is convex. 
\end{remark}
}

\subsection{Electricity System Inefficiency}

The degradation of the system's efficiency  due to selfish agents is known as the \emph{price of anarchy}  \cite{papadimitriou2001algorithms}. 
The research in this field has been focused on quantifying the efficiency loss in specific game environments \cite{tardos, johari2004efficiency, johari2005efficiency}. In this work, we are interested in the relationship between the electricity system and the \emph{tragedy of the commons} \cite{hardin1968tragedy}. 
\modification{
In the tragedy of the commons the self interest of users lead to an abuse of the resources, even if it is contrary to the interest of the whole group, because
%
agents have no incentive to individually regulate their consumption. 
The reason is that the individual benefit for abusing is greater than its cost, because the cost is shared by the community. 
Pitifully, when all individuals abuse the benefit earned by the community is reduced.
The tragedy of the commons can be defined as follows.
\begin{definition}[Tragedy of the commons]
 At the optimal outcome, every agent has incentives to consume more resources.
\end{definition}

Now, let us show that the electricity system experiences the tragedy of the commons and the implications of this fact in the design of demand response programs.
}
\modification{

\begin{theorem}\label{thm:efficiency}
 Suppose that Assumptions \ref{as:1} and \ref{as:2} are satisfied. If $\bs{\mu}$ and $\bs{\xi}$ are the solutions of the optimization problems in Eq.~(\ref{eq:opt_problem}) and Eq.~(\ref{eq:opt_problem_games}), respectively,
 then the following conditions are satisfied:
 \begin{enumerate}
 \item $\mu_i \leq \xi_i$, for all $i \in \mathcal{P}$.
  
 \item The Nash equilibrium $\bs{\xi}$ is an inefficient outcome in the sense of Pareto.
 \end{enumerate}
Hence, the game defined by $(U_1, \ldots, U_N)$ resembles the tragedy of the commons, because i) every user has incentives to increase its demand; and ii) the self interest leads to an undesired social outcome.
\end{theorem}

}

The proof for this and all other theorems and propositions are given in the Appendix.

\modification{
\begin{remark}
Without the average cost price scheme some company might be tempted to induce inefficient outcomes in the system in order to cause a greater energy demand.
\end{remark}
}

An implication of the tragedy of the commons is that, although there are plenty of resources and consumption capacity, it is not convenient to over-exploit  resources. 
In particular, the optimal outcome is characterized by a low consumption with respect to the inefficient outcome. 
Thus, the optimal outcome  reduces not only the maximum consumption (or peak), but also the overall consumption.
In consequence, using this efficiency criteria in DR schemes might seem counter intuitive if  the objective is 
to avoid peaks and to flatten the demand,
rather than lowering the total consumption. 
Hence, the DR objectives might  not be fully captured by maximizing the aggregate surplus (see Eq.~(\ref{eq:opt_problem})).

However, we can use this optimality criteria to alleviate the burden on stressed systems, without incurring in costs associated with additional generation and transmission capacity. Particularly, by reducing peaks the supplier can eliminate the most expensive generators that are used to supply peak demand  \cite{spees2007demand, baldor}.
\modification{
The following result gives a boundary of the maximum reduction of demand that can be achieved in the optimal outcome.
}
\modification{
\begin{theorem}\label{thm:ratio}
The demand of users in the efficient and inefficient systems satisfies
\begin{equation}
 \frac{ 1}{ 2  } \leq \frac{ N+1}{ 2 N } \leq \frac{ \norm{\bs{\mu}}_1 }{ \norm{ \bs{\xi} }_1 } \leq 1 .
\end{equation}
\end{theorem}

}

\modification{
From this result it is clear that for large $N$ the implementation of a DR program can reduce the total demand in  $50\%$. Particularly, if the peak demand takes place at the same time period for both the optimal and suboptimal cases, then the peak reduction is at most $50\%$.
%

In order to  analyze the improvement in the \emph{peak-to-average} (PAR) with DR we denote $\norm{ \bs{q}^k }_1 $ as the total demand in the $k\th$ time interval, with $k\in\{1, \ldots, T\}$. Thus, the PAR of both the optimal and suboptimal solutions is defined as
\begin{align}
PAR_{o} = \frac{T \norm{\bs{\mu}^{k_p}}_1}{\sum_{k=1}^T \norm{\bs{\mu}^k}_1 }
\text{  and  } 
PAR_{so} = \frac{T \norm{\bs{\xi}^{k_{p}}}_1}{\sum_{k=1}^T \norm{\bs{\xi}^k}_1 },
\end{align}
respectively.  In this case, we assume that the peak consumption is made at the time period $k_p$.
The change in the PAR can be represented by
\begin{align}\label{eq:par_ratio}
 \frac{PAR_{so}}{PAR_{o}} = \frac{ \slfrac{ \norm{\bs{\xi}^{k_{p}} }_1 }{ \norm{\bs{\mu}^{k_p}}_1 }  }
 { \slfrac{ \sum_{k=1}^T \norm{\bs{\xi}^k }_1 }{ \sum_{k=1}^T \norm{\bs{\mu}^k}_1 } }
\end{align}

We can use Theorem \ref{thm:ratio} to find the maximum improvement in the PAR.
For that, let us denote $\phi^k = \sfrac{ \norm{\bs{\xi}^k }_1 }{ \norm{\bs{\mu}^k}_1 }$, which satisfies $1\leq \phi^k\leq 2$ from Theorem \ref{thm:ratio}. 
Thus,  $\slfrac{ \sum_{k=1}^T \norm{\bs{\xi}^k }_1 }{ \sum_{k=1}^T \norm{\bs{\mu}^k}_1 }$ can be rewritten as $ \slfrac{ \sum_{k=1}^T  \phi^k \norm{\bs{\mu}^k }_1 }{ \sum_{k=1}^T \norm{\bs{\mu}^k}_1 }
$. 
This expression has the following lower bound
\begin{equation}\label{eq:lower_bound_average}
 \slfrac{ \sum\nolimits_{k=1}^T  \phi^k \norm{\bs{\mu}^k }_1 }{ \sum\nolimits_{k=1}^T \norm{\bs{\mu}^k}_1 }
 \geq
 \min_k \phi^k,
\end{equation}
which can be extracted given that 
$\sum_{k=1}^T  \phi^k \norm{\bs{\mu}^k }_1  \geq  \sum_{k=1}^T  (\min_k \phi^k) \norm{\bs{\mu}^k }_1$
holds.
Now, Eq. (\ref{eq:lower_bound_average}) can be replaced in Eq. (\ref{eq:par_ratio}) to obtain
 \begin{equation}
  \frac{PAR_{so}}{PAR_{o}} 
 \leq 
 \frac
 {\phi^{k_p}}
 {\min_k \phi^k} \leq 2. 
 \end{equation}
The rightmost inequality follows because $1\leq \phi^k\leq 2$.

If the ratio
$\phi^k$
is the same in all time periods, then the PAR has no changes in the optimal solution, regardless of the peak reduction, i.e., $\sfrac{PAR_{so}}{PAR_{o}} =1$. 
On the other hand, the PAR might be reduced if the 
ratio 
$\phi^k$
is smaller for non-peak periods. To illustrate this let us assume that $\phi^{k_p} = r$ and $\phi^{k} = r - \epsilon$ for $k\neq k_p$, $1+\epsilon \leq r\leq 2$, and $\epsilon>0$. 
Thus, the PAR ratio in Eq. (\ref{eq:par_ratio}) can be rewritten as
\begin{equation}
  \frac{PAR_{so}}{PAR_{o}} = \frac{ r \sum\nolimits_{k=1}^T \norm{\bs{\mu}^k}_1 }{ r \sum\nolimits_{k=1}^T \norm{\bs{\mu}^k}_1 - \epsilon \sum\nolimits_{k\neq k_p} \norm{\bs{\mu}^k}_1 }
  > 1.
 \end{equation}
 
 The PAR can be reduced if the demand reduction of DR is limited by consumption constraints of customers. For instance, restrictions of the form $m_i \leq q_i \leq M_i$ with $\mu_i < m_i$ or $M_i< \xi_i$ reduce the ratio between  $ \norm{ \bs{\mu} }_1 $ and $ \norm{\bs{\xi}}_1 $, 
and  might improve the PAR (see Section~\ref{sec:sim}).

}

\section{Incentives Based Mechanism}\label{sec:incentives}

Given the unsatisfactory results of the Nash equilibrium in the previous section we are motivated to design a mechanism to improve the efficiency of the equilibrium in a strategic environment.
In this case, we use {mechanism design} to find an incentives scheme that guarantees the efficiency of the Nash equilibrium \cite{mas1995microeconomic, hurwicz06}. 
First, we introduce mechanism design and present the general structure of the incentives. Then, we analyze the mechanism  as well as its properties.
\modification{
The results in this section are obtained assuming an affine unitary price function defined as $p(z) = \beta z + b$, since the generation cost can be approximated with a quadratic function \cite{wood2012power}.
}

\subsection{Incentives-Based Mechanism}

We assume that individuals in strategic environments 
are selfish and take actions that maximize their profits based on both local (or private) and global information. 
From the perspective of mechanism design, 
the rules that govern the payoff structure (and consequently the outcome of the game) can be designed by an agent, hereinafter called \emph{the principal}.  The principal might be interested in the achievement of some social goals, such as efficiency according to a particular criteria. 
Hence, the principal's problem   consists in designing the game rules that promote the desired outcome in function of the system characteristics (or economic environment).
The \emph{economic environment} $\bs{\theta} = [\theta_1, \ldots, \theta_N]$  is defined in terms of
the private information held by each agent, 
denoted by $\theta_i \in \Theta_i$, where $\theta_i$ is referred as the type of the $i\th$ agent and $\Theta_i$ is the set of all possible types. The private information $\theta_i$ is relevant to calculate the profit that each agent receives at a given outcome.

Summarizing, mechanism design consists in designing a solution system to a decentralized optimization problem with private information
\cite{hurwicz06}. 
A mechanism $\Gamma=\{ \Sigma_1, \ldots, \Sigma_N, g(\cdot) \}$ defines a set of strategies $\Sigma_i$ for each player and an outcome rule $ g:\Sigma_1 \times \ldots \times \Sigma_N \to \mathcal{O}$ that maps from the set of possible strategies  to the set  of possible outcomes $\mathcal{O}$ \cite{mas1995microeconomic}.
%
In particular, a \emph{direct revelation} mechanism (such as the  
VCG
mechanism \cite{Vic61, clarke71, groves73}) defines the set of strategies as the private information of the agents, i.e., $\Sigma_i = \Theta_i$, for all $i\in\mathcal{P}$. 
In such cases, the principal is in charge of deciding the outcome of the game based on the information sent by the agents. 
For example, in voting systems, the strategies are preferences reports made by agents, while the outcome (the selection of one candidate) might be decided using the Borda rule or plurality with elimination, among others.
Likewise, in auctions the strategies are bids and  the outcome (item allocation and its price) might be decided using a mechanisms such as the second price auction \cite{AlgorithmicG}.

We propose an indirect revelation mechanism, which does not require revelation of private information. The mechanism uses a one-dimensional message space and can be implemented in a decentralized way (see Section~\ref{sec:implementation}).
In our setting, the type of an agent is composed by her consumption preferences.
The strategy of each agent is her consumption $q_i$, and consequently, the set of all possible strategies for the $i\th$ agent is defined as $\Sigma_i\equiv S_i = \mathbb{R}_{\geq 0}$. In this case, the set of all possible outcomes $\mathcal{O}$ is composed by all the possible electricity prices, i.e., $\mathcal{O} \equiv  \mathbb{R}_{\geq 0}$. 
Note that 
the outcome rule $g(\cdot)$ of the mechanism is the price scheme used in the electricity system. 
In this case, the mechanism objective is to achieve an optimal demand profile in a strategic setting, i.e., the Nash equilibrium should be equal to the  optimal outcome. To this end, we modify the price scheme adding an incentive function  $I_i(\cdot)$ designed to align the users' profit function with the population's objective function.
The incentive function models the externality imposed by an agent on the rest of the population. The externality is the impact in prices caused by the participation of a single individual. Thus, the incentives have the form
\begin{equation}\label{eq:I_i}
I_i(\bs{\bs{q}}) = \norm{\bs{q}_{-i}}_1  \left( h_i(\bs{\bs{q}}_{-i})  - p\left( \normb{q} \right) \right),
\end{equation}
where $h_i(\bs{\bs{q}}_{-i})$ is a term that estimates the electricity price when the $i\th$ user does not take part in the electricity system. The form of this incentive is related to the price used in the VCG
mechanism \cite{AlgorithmicG}
and some  payoff functions used in potential games \cite{monderer1996potential}.  
Here, we model $h_i(\bs{\bs{q}}_{-i})$ as
\begin{equation}\label{eq:h}
h_i(\bs{\bs{q}}_{-i})=
p\left( \sum\nolimits_{j\neq i} q_j + f(\bs{\bs{q}}_{-i}) \right),
\end{equation}
\modification{
where $f(\bs{\bs{q}}_{-i})$ is a function that represents the alternative behavior 
of the $i\th$ agent. In this case, we consider linear functions of the form 
\begin{equation}\label{eq:f}
f(\bs{\bs{q}}_{-i}) = \sum\nolimits_{j\neq i} \omega_j q_j ,
\end{equation}
where $\omega_i\in \mathbb{R}$, for all $i\in \mathcal{P}$.
}

Along these lines, with the introduction of this incentives mechanism we obtain a new game 
defined as the 3-tuple  $G_I = \langle \mathcal{P}, (\Sigma_i)_{i \in \mathcal{P}}, (W_i)_{i \in \mathcal{P}} \rangle$, where $\mathcal{P}$ is the set of players, $\Sigma_i$ is the set of available strategies of each player, and $W_i: \Sigma_1 \times \dots \times \Sigma_N \to \mathbb{R}$ is the surplus function of the $i\th$ player, which is defined as
\begin{multline}
W_i(q_i,\bs{q}_{-i}) = U_i( q_i,\bs{q}_{-i} ) + I_i( \bs{q} )
\\ = v_i(q_i) -  \normb{q} p\left( \normb{q} \right) + || \bs{q}_{-i} ||_1 h_i(\bs{q}_{-i}).
\label{eq:W_i}
\end{multline}
It can be proved that the Nash equilibrium of the game $G_I$ is equal to the optimal equilibrium \modification{$\bs{\mu}$} of the original game $G$ defined in Section~\ref{sec:game} \cite{barreto2013design}.

\subsection{Mechanism Properties}
\modification{
In mechanisms design it is ideal to satisfy the budget balance property, which
states that
the net payments are equal to zero, that is, 
the sum of charges is equal to the total cost \cite{AlgorithmicG}. 
Let us define the budget balance as follows.
\begin{definition}
 A mechanism that implements payments $t_i$ is budget balanced 
 if
$
 \sum\nolimits_{i\in\mathcal{P}} t_i(\bs{q}) = C(\norm{\bs{q}}_1).
$
\end{definition}

Note that the original game with the average cost price satisfies the budget balance property, because the total payments are equal to the generation cost (see Eq. (\ref{eq:avg_price})).
However, the introduction of incentives modifies the amount charged to each costumer, and hence, the budget balance property changes.
Here, we want to determine 
if it is possible to find some function $f(\cdot)$ such that the 
rewards introduced by the incentive function $I_i(\cdot)$ can be found by the benefits obtained in the optimal outcome. 
%
%
The budget balance condition with incentives implies that 
$
 \sum\nolimits_{i\in\mathcal{P}} q_i \, p(\norm{\bs{q}}_1) - I_i(\bs{q}) = C(\norm{\bs{q}}_1).
$
Since $p(\cdot)$ is the average cost price scheme, the budget balance property is satisfied if
$
 \sum\nolimits_{i\in\mathcal{P}} I_i(\bs{q}) = 0.
$

}

In the next theorem we prove that the incentives scheme does not satisfy the budget balance property. That is, the amount of rewards (price discounts) and penalties (price increment) are not balanced, and consequently,  the mechanism requires either inflow or outflow of resources.

%
%
\begin{theorem}\label{thm:budget}
 Suppose that Assumptions \ref{as:1} and \ref{as:2} are satisfied. Also consider that
 $p(z)=\beta z + b$, where
$z\in \mathbb{R}$, $\beta>0$, and $b\geq 0$, and a population of two or more agents. Then, 
there does not exist a function $f(\cdot)$ of the form in Eq.~(\ref{eq:f}), such that the incentives mechanism described by Eqs. (\ref{eq:I_i}) and (\ref{eq:h}) satisfies the
budget balance property.
\end{theorem}
%
%
%
%
This result is an analogous to the Myerson-Satterthwaite impossibility theorem \cite{Myerson83}, which states
the impossibility of designing a mechanisms with ex-post efficiency and
without external subsidies in games between two parties. However, Theorem \ref{thm:budget} considers a nonlinear price scheme and efficiency is defined as the maximization of the aggregate surplus, rather than the maximization of the aggregate valuation.

Now, since it is not possible to find a budget balanced mechanism, 
we investigate the design of  a mechanism that satisfies the following fairness conditions.
\begin{condition}[Fairness conditions] \label{cond:incentives}
\textit{
\begin{enumerate}
 \item Incentives for the $i\th$ and $j\th$ agents are equivalent if 
their consumption is the same, i.e., if 
$q_i=q_j$, then $I_i(\bs{\bs{q}}) = I_j(\bs{\bs{q}})$. 
\item If $q_i=q_j$ for all $i,j\in\mathcal{P}$, 
then $I_i(\bs{\bs{q}}) = I_j(\bs{\bs{q}})=0$.
\item A higher power consumption deserves a lower incentive,
i.e.,  if $q_j>q_i$, then $I_j(\bs{\bs{q}}) < I_i(\bs{\bs{q}})$.
\end{enumerate}
}
\end{condition}
The following result shows the existence of a mechanism that satisfies the fairness conditions stated above.


\begin{proposition}\label{prop:f}
 Assume a population with $N\geq 2$ agents, incentives defined by Eq. (\ref{eq:I_i}) and (\ref{eq:h}), and an affine price function $p(z)=\beta z + b$ for some $\beta>0$, and $b\geq 0$. 
If the function $f(\cdot)$  has the form
\begin{equation}\label{eq:f_ideal}
f(\bs{\bs{q}}_{-i}) = \frac{1}{N-1} \sum\nolimits_{h\neq i} q_h, 
\end{equation}
for all $i\in \mathcal{P}$, then the incentives mechanism satisfies the fairness properties in Condition \ref{cond:incentives}.
\end{proposition}
%
%
%
%
Henceforth, we are going to use the following incentives that satisfy Condition \ref{cond:incentives}: 
 \begin{equation}\label{eq:I_i_final}
  I_i(\bs{q}) = \norm{ \bs{q}_{-i} }_1 \left( 
   p\left( \frac{N}{N-1} \norm{  \bs{q}_{-i} }_1 \right) - 
   p\left( \normb{q} \right) \right).
 \end{equation}
Therefore, the surplus function in Eq. (\ref{eq:W_i}) can be rewritten as
 \begin{equation}\label{eq:u_incentives}
  W_i(\bs{q}) = v_i(q_i) -  \normb{ q} p\left( \normb{q} \right) + \norm{ \bs{q}_{-i} }_1 p\left(  \frac{N}{N-1} \norm{ \bs{q}_{-i} }_1 \right).
 \end{equation}

 \modification{
\begin{remark}
 Each user only needs the aggregate demand $||\bs{q}_{-i}||_1$ to calculate the consumption $q_i$ that maximizes $W_i(\bs{q})$. 
   Thus, the mechanism can be implemented using a one-dimensional message space that communicates the aggregate demand to each user. When considering $T$ time periods, we can still use a one-dimensional message space if the consumption $q_i^k$ is calculated sequentially. Otherwise, the implementation might require a $T-$dimensional message space to calculate simultaneously the consumption along a day. 
\end{remark}

}

With an affine price function the incentives can be rewritten as
\begin{equation}\label{eq:I_i_ideal}
I_i(\bs{\bs{q}}) = \beta \Big( \sum\nolimits_{j\neq i} q_j \Big) 
\Big( \frac{1}{N-1} \sum\nolimits_{j\neq i} q_j - q_i \Big)
\end{equation}
for all $q_i\geq0$, $i,j \in \mathcal{P}$. Thus, the population
incentives can be expressed as
\begin{equation}\label{eq:I_i_ideal_matrix}
\sum\nolimits_{i\in\mathcal{P}} I_i(\bs{\bs{q}}) = \beta \bs{q}^\top A \bs{q},
\end{equation}
where $A=\Big( \frac{-1}{N-1} \bs{e} \bs{e}^\top + \frac{N}{N-1} I\Big)$ and $\bs{e}$ is a vector in $\mathbb{R}^N$ with all its components equal to 1.
 Now, with this expression we can  analyze some properties of incentives given by Eq.~(\ref{eq:I_i_ideal}). 
 In the next proposition we show that the system requires external subsidies to maintain the incentives scheme. In other words, 
the mechanism has budget deficit \cite{mas1995microeconomic}.
\begin{proposition}\label{prop:subsidies}
 Suppose that Assumptions \ref{as:1} and \ref{as:2} are satisfied.
Given an incentives mechanism of the form in Eq. (\ref{eq:I_i_ideal}), then the incentives required by the population are positive, i.e., 
$\sum_{i\in\mathcal{P}} I_i(\bs{\bs{q}}) \geq 0,$
for all $\bs{q}\in\mathbb{R}_{\geq 0}^N$,
\end{proposition}
%
%
%
%

\begin{remark}
Let us consider an homogeneous population, composed by agents with equal preferences. In such population, the energy consumed at the equilibrium is the same for every agent. Thus, according to Condition \ref{cond:incentives}, 
a homogeneous population requires null incentives at the equilibrium. In particular, incentives would be required only to shift the system from an inefficient outcome toward the optimal equilibrium.
\end{remark}

\modification{
It is possible to design incentive functions that lead to systems with \emph{budget surplus} (revenues are higher than expenses). For example, \cite{barreto_acc15} shows an alternative incentive function that entails higher electricity prices  and specifically, does not satisfies $\textbf{ii})$ in Condition \ref{cond:incentives}, because the incentives are negative even when all users have equal demand. 

Now, using Proposition \ref{prop:subsidies} we can prove that the aggregate surplus reached in the Nash equilibrium with incentives is equal to the aggregate surplus of the optimal solution $\bs{\mu}$. 
}
\begin{proposition}\label{prop:improvement}
 Consider a population of agents with surplus function of the form in Eq.~(\ref{eq:u_i_}) and incentives described by Eq.~(\ref{eq:I_i}). Also, consider that Assumptions \ref{as:1} and \ref{as:2} are satisfied. Then, the aggregate surplus with the mechanism is equal to the maximum aggregate surplus of the initial game $G$, i.e., 
   $\sum\nolimits_{i\in\mathcal{P}} W_i (\bs{\mu})  - I_i(\bs{\mu})= \sum\nolimits_{i\in\mathcal{P}} U_i( \bs{\mu} ) 
  > \sum_{i\in\mathcal{P}} U_i( \bs{\xi} ).$
 \end{proposition}

Next, we prove that the incentives mechanism is individual rational.
\modification{
A mechanism is individual rational if agents voluntarily accept the rules imposed by the mechanism. 
Formally, a mechanism is individual rationality
if the  benefit obtained by any agent $i\in\mathcal{P}$ 
with the mechanism
is greater than the benefit of not participating, which is assumed to be zero. 
%

%
\begin{theorem}\label{thm:rationality}
 The mechanisms with incentives in Eq.~(\ref{eq:I_i_final}) is individual rational, that is, 
 $
  W_i(\mu_i,\bs{\mu}_{-i}) \geq 0,
 $
 for all $\bs{\mu}$.
\end{theorem}

}

 So far, we have verified that in a strategic environment the aggregate surplus is optimal with the adoption of incentives (see Proposition~\ref{prop:improvement}). Also, the mechanism guarantees that the surplus of an individual is always positive. However, an individual that is enrolled in an inefficient system might migrate toward a system that implements incentives only if its  profit is not be reduced after the change. 
 The following result guarantees that the agents that have a lower consumption with respect to the average  can expect a greater surplus in the system with incentives.

 \begin{theorem}\label{thm:benefit_agents}
  Every agent that consumes less resources than the average in the optimal outcome of the game $G$
  can expect a greater profit in the Nash equilibrium of the game with incentives $G_I$.
  That is, 
  if $\mu_i < \frac{1}{N} \sum_{h\in\mathcal{P}} \mu_h$, then $U_i(\bs{\mu}) < W_i(\bs{\mu})$. Otherwise, 
  $U_i(\bs{\mu}) \geq W_i(\bs{\mu})$.
 \end{theorem}
%


Theorem \ref{thm:benefit_agents} shows that some users can expect a higher surplus in the system with incentives ($G_\mathcal{I}$).
This happens because the low consumption is rewarded in $G_\mathcal{I}$.
The extent to which an agent can expect major surplus in $G_\mathcal{I}$, with respect to the inefficient outcome $\bs{\xi}$ is an open problem.

\modification{
In summary, a careful design of the incentives properties is necessary to encourage the adoption of the mechanism. 
For instance, we know that customers would join the DR program because it guarantees positive surplus (see Theorem \ref{thm:rationality}). 
Furthermore, if it is their choice, some customers would prefer a program with incentives (see Theorem \ref{thm:benefit_agents}). 
These properties can be assured partly because the mechanism has budget deficit. On the contrary, if the mechanism is weak budget balanced, then  customers would have to face higher taxes. 
}

\section{Decentralized Implementation of the Mechanism}\label{sec:implementation}

\begin{figure}[bt]
 \centering
 \includegraphics[width=.45 \textwidth]{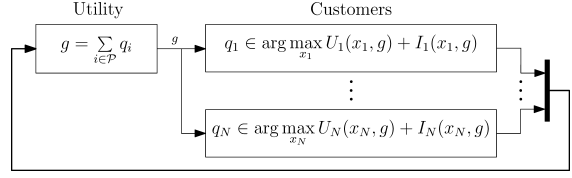}
 \caption{Decentralized implementation of the mechanism. Each agent must compute its optimal consumption profile $q_i$ in function of the aggregate consumption profile $g$.}
 \label{fig:decentralized}
\end{figure}

In previous sections we have analyzed the characteristics of the electricity system market at different equilibrium points. 
In this section  we are concerned with the behavioral modeling (dynamics) of rational individuals that are involved in a game. 
Particularly, we show that the Nash equilibrium can be learned in a decentralized manner.

\modification{
Let us modify the notation used in previous sections to allow multiple time periods.
We divide a period of 24 hours into $T$ disjoint time intervals, denoted by $\pi_1,\ldots,\pi_T$,
that satisfy $\cup_{k \in\{1,\ldots,T\}} \pi_k = [0,24)$ and $\cap_{k\in\{1,\ldots,T\}} \pi_k = \varnothing$. The consumption of the $i\th$ customer in the $k\th $ time interval is denoted by  $q_i^k$.
Thus, the daily consumption profile of a customer is represented with the vector $\bs{q}_i=[q_i^1,\ldots, q_i^T ] ^\top \in \mathbb{R}_{\geq 0}^T$.
Likewise, the electricity consumption of the population is denoted by the 
vector $\bs{q} = [\bs{q}_1^\top, \ldots, \bs{q}_N^\top]^\top \in \mathbb{R}_{\geq 0}^{T\cdot N}$, 
and the vector
$\bs{q}^{k} = [ q_1^k, \ldots, q_N^k ]^\top$ represents the electricity consumption of all customers at the $k\th$ time interval.
On the other hand, the vectors 
$\bs{q}_{-i} = [\bs{q}_1^\top, \ldots, \bs{q}_{i-1}^\top, \bs{q}_{i+1}^\top, \ldots,  \bs{q}_N^\top]^\top \in \mathbb{R}_{\geq 0}^{T\cdot (N - 1)}$ 
and
$\bs{q}_{-i}^{k} = [ q_1^k, \ldots, q_{i-1}^k, q_{i+1}^k, \ldots,  q_N^k ]^\top$
represent the consumption of the population without the $i\th$ agent 
during a day or during the $k\th$ time period, respectively.
The valuation function $v_i^k:\mathbb{R}_{\geq0}\rightarrow \mathbb{R}$ represents the economic value that the $i\th$ user assigns to $q_i^k$ electricity units in the $k\th$ time interval. 
}

We propose an indirect revelation mechanism that 
allows each customer to use the resources without any restriction.  
In this way it distributes the computation tasks among the population. 
Specifically, each individual is responsible for optimizing its own consumption based on the total demand of the society, that can be  broadcasted by the electric utility 
(see Fig.~\ref{fig:decentralized}).

We can think on automation devices that optimize the energy consumption (following some dynamics), based on particular preferences (local information) and reports from the central entity.
%
Therefore, we assume that each customer's automation device carries out a learning process to adjust its consumption to the 
price signals sent by the electric utility.
In other words, the learning process solves a resource allocation problem, in which each individual finds the amount of resources that should be used 
in a given time period.
Here, we assume that the daily consumption of each user is bounded by $Q_i$, which can be interpreted as the maximum consumption capacity of the $i\th$ customer.
Let us generalize the surplus function 
\modification{
with incentives in Eq. (\ref{eq:W_i}) as
\begin{multline}
W_i (\bs{q}_i,\bs{q}_{-i}) =
\sum\nolimits_{k=1}^T  \left( 
v_i^k(q_i^k) -  \norm{\bs{q}^k}_1 \, p\left( \norm{\bs{q}^k}_1 \right) + \right.
\\
 \left. \norm{\bs{q}_{-i}^k}_1 \, h_i(\bs{q}_{-i}^k)
\right)
\end{multline}
}
Accordingly, the optimization problem that customer agent solves (see Eqs.~(\ref{eq:opt_problem_games}) and (\ref{eq:u_incentives})) can be rewritten as
\begin{equation}\label{eq:opt_problem_games_pop}
\begin{aligned}
& \underset{\bs{q}_i}{\text{maximize}}
& & W_i (\bs{q}_i,\bs{q}_{-i}) \\
& \text{subject to}
& & \sum\nolimits_{k=1}^{T} q_i^k  \leq Q_i \\
& & & q_i^k \geq 0,  i =\{1,\ldots,N\}, k =\{1,\ldots,T\}.
\end{aligned}
\end{equation}
If $Q_i$ is large enough, we can assure that the solution to the optimization problems in  (\ref{eq:opt_problem_games}) and (\ref{eq:opt_problem_games_pop}) is the same. 
This formulation is convenient to define the fitness and strategies in the population game defined below. Now, let us introduce some notation to be consistent with the literature in population games \cite{sandholm_book}.

\begin{remark}
 This model does not has into account other factors that impact the satisfaction of users, such as quality (e.g., continuity of the service, variation in voltage), reliability, and security of the service. The regulation made by the ISO encourages high quality of service through economic incentives, and therefore imposes additional restriction on the electric utility (which are not considered here). We refer the reader interested in security aspects to \cite{BarretoACSAC14}, where it is shown that this decentralized scheme might be more resilient to fraud and malicious attacks than centralized schemes (such as direct load control).
\end{remark}

\subsection{Population Games Approach}

In this case, we assume that each user implements some evolutionary dynamics to maximize her own surplus (see problem in Eq. (\ref{eq:opt_problem_games_pop})). 
Thus, the electricity game can be seen as a multi-population game, in which each customer represents a population in the society $\mathcal{P}$.
%
%
%
%
The evolutionary dynamics are differential equations that describe changes in the strategies adopted in the population (in this case the demand along the day).
%
%
%
The population dynamics approach can be used to solve optimization problems with restrictions, as it is shown in \cite{pantoja}. 

Let us formulate the population game as follows.
We consider a society composed by $N$ populations with $T+1$ possible strategies per each population.
For a given population $i$, the $k\th$ strategy's expected value is denoted by 
$q_i^k$,
for $k \in \{1,\ldots,T\}$, i.e., the population's strategy is the amount of resources consumed in each time period. 
Moreover, the strategy 
$q_i^{T+1}$ 
is a slack variable that represents the  power not consumed in any time interval  and it is modeled as a consumption in the fictitious $(T+1)\th$ time interval. 
The slack variable is defined as
\begin{equation}
q_i^{T+1} = Q_i - \sum\nolimits_{k=1}^T q_i^k.
\end{equation}
Now, let us define the fitness (or payoff) function $F_k^i:\mathbb{R}^N \to \mathbb{R}$ for the $k\th$ strategy in the $i\th$ population as the derivative of the surplus function $W_i(\bs{q})$,  i.e.,  the fitness is equal to the marginal surplus of the $i\th$ population,  defined as
\begin{equation}\label{eq:fitness}
F_k^i( \bs{q}^k ) = \frac{\partial W_i}{\partial q_i^k}(\bs{q}^k) ,
\end{equation}
for $k \in \{1,\ldots,T\}$. 
On the other hand, the fitness of the fictitious variable (consumption in $k=T+1$) is defined as zero, i.e., $F_{T+1}^i = 0$.
The election of these fitness functions guarantees that the population game is a potential game.
Potential games are a class of population games that can be solved using multiple dynamics \cite{sandholm_book}. 
\modification{
\begin{proposition}
The proposed population game is a potential game with potential function $\varPsi(\bs{q}) = \sum_{i\in\mathcal{P}} U_i(\bs{q}_i, \bs{q}_{-i}) $.  
%
%
\end{proposition}
}

\modification{
In this case, we use the definition of potential games for continuous sets of strategies defined in \cite{sandholm_book}. Specifically, a population game with fitness functions $F_k^i$ is potential game if there exist a function $\varPsi$ such that
\begin{equation}
\frac{\partial \varPsi}{\partial q_i^k} (\bs{q})=F_k^i(\bs{q}^k).
\end{equation}
The main property of potential games is that the Nash equilibria of the game is characterized by the Karush-Kuhn-Tucker first order necessary conditions of the potential function. This implies that the incentives of all players are mapped into one function and the set of pure Nash equilibria can be found by locating the local optima of the potential function.
}

\modification{
The dynamics are built assuming myopic behavior, hence, the decisions are made using past information of the aggregated demand. For instance, users choose the future consumption at a particular time period based on the consumption from the day before. However, the updates for the whole day must be calculated at the same time, because the dynamics couple different time periods. 
}

\subsection{Evolutionary Dynamics} \label{sec:dynamics}

We implement 
four evolutionary dynamics, namely \emph{logit dynamics} (Logit), \emph{replicator dynamics} (RD), \emph{Brown-von Neumann-Nash dynamics} (BNN), and 
\emph{Smith dynamics}, which belongs to the family of
\emph{perturbed optimization}, \emph{imitative dynamics}, \emph{excess payoff dynamics}, and 
 \emph{pairwise comparison dynamics} \cite{hofbauer2001nash, sandholm_book}. 
The following differential equation describe the evolution in time of each strategy 
\subsubsection{Logit Dynamics}
\begin{equation}\label{eq:logit}
 \dot{x}_k^i = \frac{ \exp\left(\eta^{-1} F_k^i (\bs{x}) \right) }{ \sum_{\gamma \in S} \exp\left(\eta^{-1} F_\gamma^i (\bs{x}) \right) }, \, \, \eta>0,
\end{equation}

\subsubsection{Replicator Dynamics}
\begin{equation}\label{eq:replicator}
\dot{x}_k^i = x_k^i \, \hat{F}_k^i \left( \bs{x} \right).
\end{equation}

\subsubsection{Brown-von Neumann-Nash Dynamics (BNN)}
\begin{equation}\label{eq:bnn}
 \dot{x}_k^i = \left[ \hat{F}_k^i \left( \bs{x} \right) \right]_+ - x_k^i  \sum\nolimits_{\gamma \in S} \left[ \hat{F}_\gamma^i \left( \bs{x} \right) \right]_+
\end{equation}
\subsubsection{Smith Dynamics}
\begin{multline}
\dot{x}_k^i  = \sum\nolimits_{\gamma \in S} x_\gamma^i  \left[ F_k^i \left( \bs{x} \right) - F_\gamma^i \left( \bs{x} \right) \right]_+ 
\\
- x_k^i  \sum\nolimits_{\gamma \in S} \left[ F_\gamma^i ( \bs{x}) - F_k^i( \bs{x} ) \right]_+.
\label{eq:smith}
\end{multline}

This dynamic is  defined in function of the excess payoff to strategy $k$ as $\hat{F}_k^i =  F_k^i(\bs{q}^k) - \bar{F}_k^i(\bs{q}^k)$, where $\bar{F}_k^i(\bs{q}^k)$ is the average payoff the population $i$.
Since the potential function $\varPsi(\cdot)$ is a concave function, we know 
that the population game has a unique Nash equilibrium, which corresponds to the maximum of $\varPsi$. 
The dynamics in Eq.~(\ref{eq:smith}) satisfies the \emph{positive correlation} (PC) property. 
Hence, according to Lemma 7.1.1 \cite{sandholm_book}, $\varPsi$ is a Lyapunov function for the differential equation and it can be shown that the Nash equilibrium is globally asymptotically stable with Smith dynamics. 
Moreover, Replicator dynamics is locally stable because it does not necessarily converge to the Nash equilibrium. 
In particular, the solutions to (\ref{eq:replicator}) might not reach the Nash equilibrium if the initial conditions are in the border of the simplex $X^i$ (replicator admits solutions/rest points that are not NE as well as  closed orbits). \modification{Analogous}  convergence results can be derived for the perturbed best response (logit) dynamics.

It is important to highlight that in this implementation we use two different time domains. 
On the one hand, we represent a daily time domain by means of $k\in\{1,\ldots,T\}$. 
This time domain is discrete and represents different time intervals during a day. 
On the other hand, the evolutionary dynamics introduce a continuous time domain related to the evolution of the differential equations. 
The scale of this continuous time domain can be considered much larger than the daily time domain, since adjustments in the consumption are considered to be slow.

\section{Simulation Results}\label{sec:sim}

In this section, we illustrate some ideas of efficiency and the decentralized implementation of the incentives mechanism.
\modification{
Part of the simulations are available at \cite{toolbox}.
}
In these experiments we select some functions used previously in the literature   that satisfy Assumptions  \ref{as:1} and \ref{as:2} (see \cite{ mitter_2011, mohsenian}).
On the one hand, we define the valuation functions as 
\begin{equation}\label{eq:valuation_sim}
 v_i^k (q_i^k) = \alpha_i^k \log(1+q_i^k)
\end{equation}
where $\alpha_i^k>0$ is the parameter that characterizes the valuation of the  $i\th$ agent at the $k\th$ time instant.
On the other hand, the generation cost function is defined as 
\begin{equation}\label{eq:cost_sim}
 C(\|\bs{q}\|_1) = \beta ({\|\bs{q}\|_1})^2, 
\end{equation}
and the unitary price function is
\begin{equation}\label{eq:p_sim}
 p(\|\bs{q}\|_1) = \frac{C(\|\bs{q}\|_1)}{\|\bs{q}\|_1} = \beta \|\bs{q}\|_1. 
\end{equation}
\modification{
For simplicity,
}
the generation cost only depends on the aggregate consumption, not on the time of the day. Furthermore, 
the fitness function of the system with incentives (see Eq. (\ref{eq:fitness})) is
\begin{equation}\label{eq:fitness_without_i_sim}
F_k^i( \bs{q}^k)  =  \frac{\alpha_i^k }{1+q_i^k}
 - 2\beta \left( \sum\nolimits_{j\in\mathcal{P}} q_j^k  \right).
\end{equation}

We define $N=5$ users, $Q_i=30 KWh$ for all $i$, $\beta = 1$,  $T = 24$, and random initial conditions that satisfy $\sum_{k=1}^{T+1} x_k^i (0) = Q_i$. 
%
In order to model time varying valuations along a day, we assign to $\alpha_i^k$ a value proportional to the actual consumption in an electrical system. 
\modification{
In this case, we 
use the consumption measurements provided by the Colombian electricity system administrator to choose $\alpha_i^k$ proportional to the total demand during the $k\th$ time period 
\cite{xm}.}
We define a heterogeneous society composed by individuals with different valuations, such that
\begin{equation}
\alpha_i^k < \alpha_j^k,    
\end{equation}
for all $i,j\in\mathcal{P}$ with \modification{$i < j$} and $k\in\{1,\ldots,T\}$. Thus, the $i\th$ user has lower preferences, at any time interval, than the ${i+1}\th$ user.


\subsection{Inefficiency Example}

In Fig.~\ref{fig:surplus} we show the social surplus and the total demand of the optimal and suboptimal  solutions of the game $G$ (without incentives). We verify that the population's demand is lower at the optimal solution. Also, Fig.~\ref{fig:surplus} shows that the surplus of the society is greater at the optimal outcome. These properties are characteristics of the tragedy of the commons.
\begin{figure}[bt]
 \centering
 \includegraphics[width=.45\textwidth]{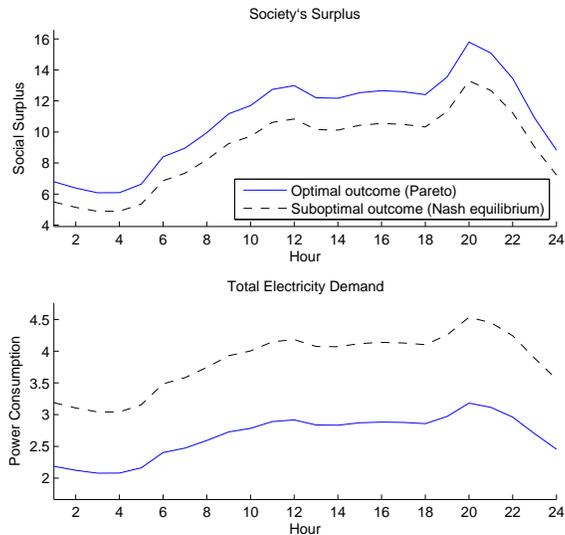}
 \caption{Social surplus and total demand  in both optimal and suboptimal solutions.
 }
 \label{fig:surplus}
\end{figure}
%

\modification{
Now we are interested in finding the relation between the optimal and suboptimal solutions as a function of the number of customers $N$. Fig. \ref{fig:ratio} shows the change in both demand and surplus of the optimal solution with respect to the suboptimal case for the peak hour, which takes place at the same time in both solutions.
As $N$ grows the ratio $\sfrac{\big\lVert\bs{\mu}^{k}\big\lVert_1}{\big\lVert\bs{\xi}^{k}\big\lVert_1}$ approaches $\sfrac{1}{2}$. In other words, the optimal demand is at least a half of the suboptimal demand. 
Moreover, the quotient of the social surplus in the optimal and suboptimal solutions increases with $N$. The efficiency loss boundary presented in \cite{johari2006scalable} does not apply here because we use the average cost price scheme, rather than marginal cost price.
}

\begin{figure}[tb]
        \centering

        \begin{subfigure}[b]{.45\textwidth}
                \centering
		  \includegraphics[width=\textwidth, clip=true]{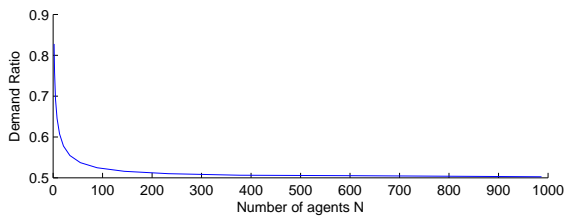}
                   \caption{Quotient of the total demand in the optimal and suboptimal solutions.}
		 \label{fig:demand_ratio}
        \end{subfigure}%
        
        \begin{subfigure}[b]{.45\textwidth}
                \centering
		  \includegraphics[width=\textwidth]{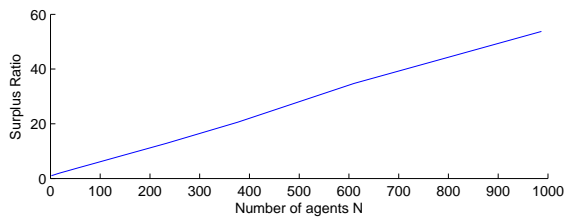}
                   \caption{Quotient of the social surplus in the optimal and suboptimal solutions.}
		 \label{fig:surplus_ratio}
        \end{subfigure}%

        \caption{ Relation of the total demand and the social surplus in both the optimal and suboptimal solutions.}
	 \label{fig:ratio}
\end{figure}

\modification{
%
Recall from Section~\ref{sec:backgnd} that the ratio $\sfrac{\big\lVert\bs{\mu}^{k}\big\lVert_1}{\big\lVert\bs{\xi}^{k}\big\lVert_1}$ might be affected by restrictions of the form $m_i \leq q_i^k \leq M_i$. Here we variate $m_i$ to observe the change in the PAR ratio from Eq. (\ref{eq:par_ratio}). Fig. \ref{fig:par_restrictions} shows that the the PAR ratio improves if $\min_k \norm{\bs{\mu}^k}_1 \leq m_i \leq \max_k \norm{\bs{\mu}^k}_1$, because in this case the PAR of the optimal solution is reduced, while the PAR of the sob-optimal solution is the same. However, if 
$m_i \geq \min_k \norm{\bs{\xi}^k}_1$ then the PAR ratio decreases, because the restriction affects the sub-optimal solution and reduces its PAR.
}

  \begin{figure}[t]
    \centering
    \begin{tikzpicture}
      \node[anchor=south west, inner sep=0pt] (fig) at (0, 0)
      { \includegraphics[width=.45\textwidth, clip=true]{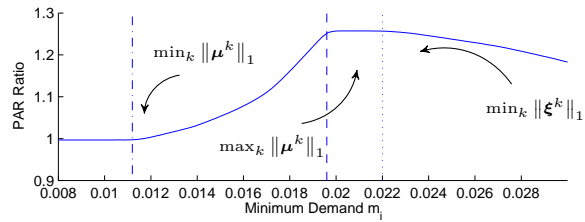} };
      \node[] (a) at (2.6, 2.25) {\tiny $\min_k \norm{\bs{\mu}^k}_1$};
      \node[] (a_end) at (1.81, 1.25) {}
      edge[pil, bend left=40] (a.225);
      
      \node[] (b) at (3.5, 1.0) {\tiny $\max_k \norm{\bs{\mu}^k}_1$};
      \node[] (b_end) at (4.65, 2.2) {}
      edge[pil, bend left=35] (b.40);

      \node[] (c) at (7, 1.5) {\tiny $\min_k \norm{\bs{\xi}^k}_1$};
      \node[] (c_end) at (5.28, 2.2) {}
      edge[pil, bend left=35] (c.135);      
      
    \end{tikzpicture}  
 \caption{PAR improvement reached when considering restrictions of the form $m_i \leq q_i^k$ with different values of $m_i$.}
 \label{fig:par_restrictions}
  \end{figure}

Now, let us analyze the ideas related with the dynamical systems. 
In order to the analyze the response of the population to the economic incentives, we introduce  incentives in the time period \modification{contained} between 2 and 4 seconds (see Fig. \ref{fig:dynamics}). 
The introduction of the incentives produces an increment in population average utility, 
as well as a reduction in the average consumption. Also, note that the total incentives delivered to the
heterogeneous population are different from zero. 
%

\begin{figure}[hbt]
 \centering
 \includegraphics[width=.5\textwidth]{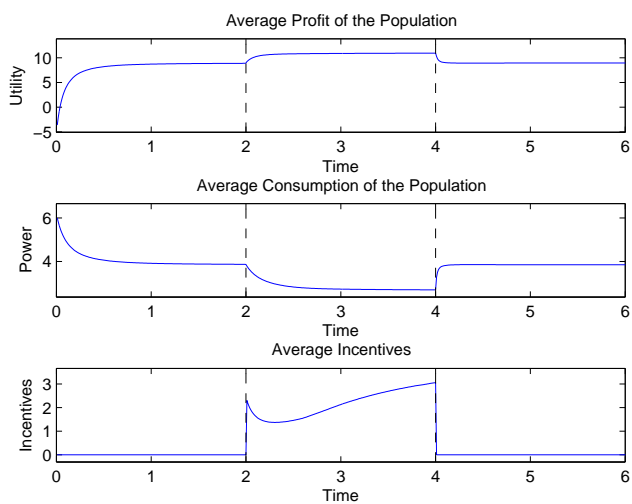}
 \caption{Dynamics of a heterogeneous population with incentives using BNN dynamics.}
 \label{fig:dynamics}
\end{figure}

Fig.~\ref{fig:incentives_eq} shows the incentive given to each user in the optimal equilibrium. Note that the user with a lower valuation (user 1) is the one that receives more incentives, while the user with larger valuations (user 5) has the lower incentives. This happens because the user with lower valuations can reduce its consumption much more than a user with higher valuations, and consequently,  receives more incentives.

\begin{figure}[hbt]
 \centering
 \includegraphics[width=.45\textwidth]{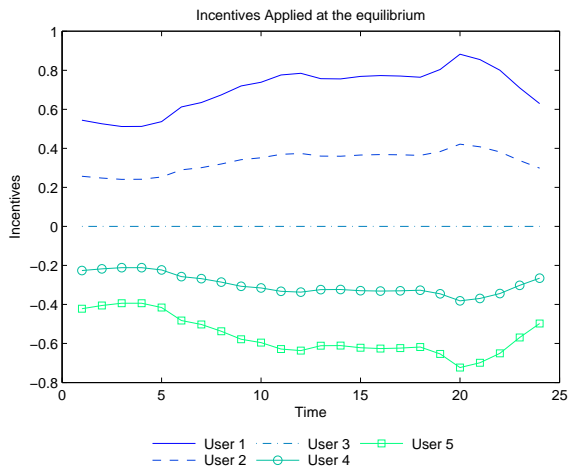}
 \caption{Incentives of a heterogeneous population at the optimal equilibrium.}
 \label{fig:incentives_eq}
\end{figure}

\subsection{Evolutionary Dynamics}

\begin{figure}[bt]
 \centering
 \includegraphics[width=.45\textwidth]{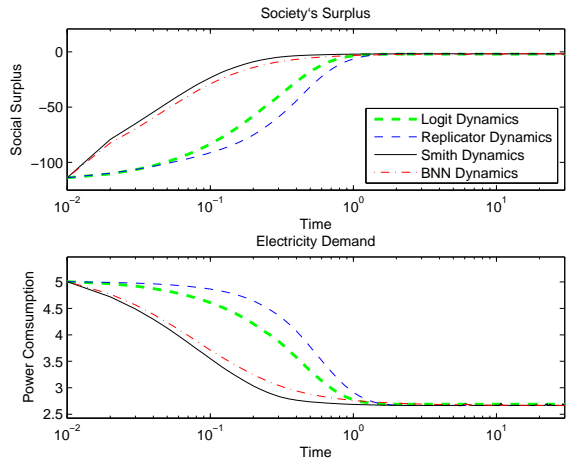}
 \caption{Evolution of the society's daily average surplus and power consumption.}
 \label{fig:dynamics_u}
\end{figure}

\begin{figure}[bt]
 \centering
 \includegraphics[width=.45\textwidth]{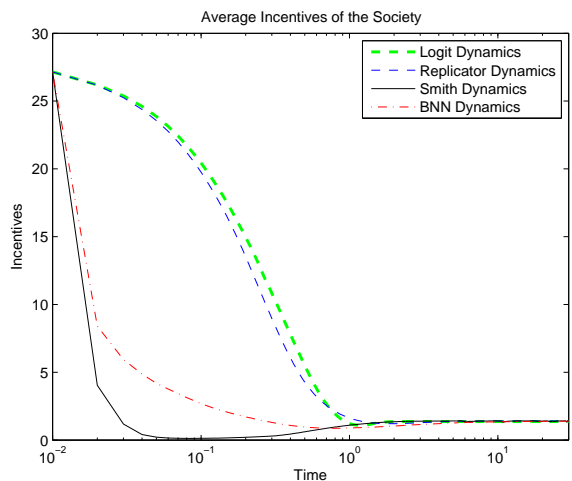}
 \caption{Evolution of the incentives with four different dynamics.}
 \label{fig:dynamics_i}
\end{figure}



Now, let us analyze the ideas related with the dynamical systems. 
The evolution of social surplus, demand, and incentives for different dynamics are shown in Figs.~\ref{fig:dynamics_u} and \ref{fig:dynamics_i}. Note that despite using the same initial condition, the evolution of the system is different with each dynamical model. In particular, BNN and Smith dynamics converge faster to the  optimum.
This is achieved by means of a fast decrease in the power consumption. 

Incentives in Fig.~\ref{fig:dynamics_i} show that, in the long run, all dynamics converge to the same level of incentives. Particularly, Smith dynamics requires more incentives during all time, except for logit dynamics, which has a sudden increase in the incentives close to the equilibrium point. 

In Fig.~\ref{fig:dynamics_i} it is not clear which dynamical model moves the state of the system to the optimal equilibrium using less resources. To answer this question, we simulate the total amount of incentives used by each model.
Thus, let us define the aggregate incentives in a society in a particular time $t$ as
\begin{equation}
 I_d (t) = \sum\nolimits_{i\in\mathcal{P}} \frac{1}{|S|} \sum\nolimits_{k\in S} I_i \left( \bs{q}^k (t) \right).
\end{equation}
Now, the total accumulated incentives from $t_0$ to $t$ is defined as 
\begin{equation}
 \varPhi_d (t) = \int_{t_0}^t I_d (\tau) d\tau.
\end{equation}
Thus, $\varPhi_d (t)$ gives a measurement of the total amount subsidies required by the system with dynamic $d$, in the time interval  $[t_0, t]$.
In this case we do not have a reference to compare the subsidies requirements of each evolutionary dynamic. Hence, we compare the subsidies requirements with the average requirements of all the dynamics implemented. 
In order to see which dynamic requires more resources, we plot the cumulative resources required by each dynamic relative to the average.
Hence, we define the cumulative incentives as 
\begin{equation}
CI_d = \frac{ \varPhi_d (t) }{ \sum_{d\in \mathcal{D}} \varPhi_d (t) }.
\end{equation}
Fig.~\ref{fig:integral} shows the results of the simulation of the relative subsidies required by each model of evolutionary dynamics.
Smith and BNN dynamics require less resources, while logit has the higher requires higher incentives. However, BNN has the lower incentives in long run.

\begin{figure}[bt]
 \centering
 \includegraphics[width=.45\textwidth]{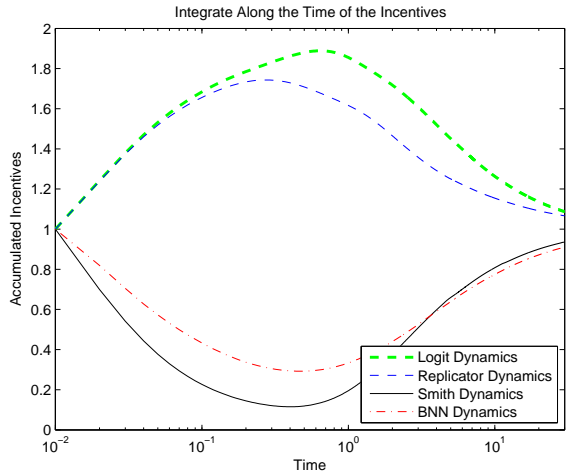}
 \caption{Accumulated incentives during the evolution of the algorithm.}
 \label{fig:integral}
\end{figure}


\section{Conclusions and Future Directions}
\label{sec:concl}

In this paper we propose an indirect revelation mechanism to maximize the aggregate surplus of the population. The main feature of this mechanism is that it does not require private information from users, and employs a one-dimensional message space per resource to be allocated.
These properties facilitate the distributed implementation of the mechanism. The mechanism entrusts the computation tasks among users, who should maximize its own surplus function based the aggregate demand (that is calculated and broadcasted by a central agent). Thus, users avoid revelation of private information (e.g., preferences), but are required to report the aggregate consumption of their appliances during some time periods.

We show that most users of the electricity system would join the incentives program voluntarily, since they might have positive surplus. Particularly, users that consume less resources than the average can expect a higher surplus in the system with incentives, because the low consumption is rewarded in $G_\mathcal{I}$. However, the extent to which an agent can expect major surplus in $G_\mathcal{I}$, with respect to the inefficient outcome $\bs{\xi}$ is an open problem.
The mechanism \modification{has budget deficit} and might require  external subsidies to succeed. In particular, the implementation cost of the mechanism depends on preferences and dynamics implemented by users. 

We introduce an approach based on evolutionary dynamics that might be used by each user to find the demand profile that maximizes its surplus. Particularly,  when implemented locally by each user, the evolutionary dynamics lead to the global efficient equilibrium.
We implement four popular evolutionary dynamics, namely logit dynamics, replicator dynamics, Brown-von Neumann-Nash dynamics, and Smith dynamics. 
We find that the system might converge faster to the equilibrium with Smith and BNN dynamics. Also, BNN dynamics has a relatively fast convergence and uses less resources in the long run.

Future work will be focused on analyzing the characteristics of the mechanism on large populations. Also, it is interesting to explore different dynamics that lead to minimum accumulated incentives, as well as possible applications of fast dynamics in environments with random components, such as renewable generation.
%
%
\modification{ 
\subsection{Limitations}

We do not consider time interdependencies of the electricity usage. Also, we assume that the customers can choose a continuous electricity usage. However, appliances have discrete consumption and additional operation constraints. This model does not has into account other factors that impact the satisfaction of users, such as quality (e.g., continuity of the service, variation in voltage), reliability, and security of the service. 

The  implementation might have a slow convergence because we omit the negotiation phase to preserve privacy. This might limit the use of renewable generation.

}

\section{Appendix}

\modification{
\begin{proof}[Proof of Theorem \ref{thm:efficiency}]
 The  proof of numeral (\textbf{i}) is made by contradiction. Let us assume that there exist some agent $j\in\mathcal{P}$ such that $\mu_j > \xi_j $.
 From Eq. (\ref{eq:FOC2}) we know that the Nash equilibrium $\bs{\xi}$ satisfies the following FOC:
 \begin{equation}\label{eq:foc_nash}
  \frac{\partial }{\partial q_j} U_j (\bs{q}) \Big|_{\bs{q} = \bs{\xi}} 
  = \dot{v}_j(\xi_j) - p(\norm{\bs{\xi}}_1) 
  - \xi_j \, \dot{p}(\norm{\bs{\xi}}_1) = 0.
 \end{equation}
On the other hand, our initial hypothesis implies that at the Nash equilibrium $\bs{\xi}$ the $j\th$ user has incentives to increase its consumption, that is, 
\begin{multline}
 \frac{\partial }{\partial q_j} \sum\nolimits_{i\in\mathcal{P}} U_i (\bs{q}) \Big|_{\bs{q} = \bs{\xi}} = 
 \dot{v}_j(\xi_j) - p(\norm{\bs{\xi}}_1) 
 \\ 
 - \norm{\bs{\xi}}_1 \dot{p}(\norm{\bs{\xi}}_1) > 0. \label{eq:foc_pareto}
\end{multline}
Replacing Eq. (\ref{eq:foc_nash}) into Eq. (\ref{eq:foc_pareto}) we obtain
\begin{equation}\label{eq:neg_ineq}
 - \norm{\bs{\xi}_{-j}}_1 \dot{p}(\norm{\bs{\xi}}_1) > 0.
\end{equation}
Note that with Assumption \ref{as:2} we guarantee that $\mu_i \geq 0$ and $\xi_i \geq 0$, for all  agent $i\in\mathcal{P}$. Hence, 
from Eq. (\ref{eq:neg_ineq}) we conclude that $\dot{p}(\norm{\bs{\xi}}_1) < 0$. However, from Assumption \ref{as:1} we know that the unitary price function $p(\cdot)$ is an increasing function. Hence,  $\dot{p}(\norm{\bs{\xi}}_1)$ is positive or equal to zero. This leads to a contradiction, showing that $\mu_j \leq \xi_j$, for all $j\in\mathcal{P}$.

Now, the  proof of numeral (\textbf{ii}) is made by  direct proof. 
With Assumption \ref{as:1} we guarantee that the competitive equilibrium 
is unique, and corresponds to the best possible outcome for the population. Hence,
the competitive equilibrium  is efficient in the sense of Pareto.
On the other hand, from numeral (\textbf{i}) we conclude that 
$\|\boldsymbol{\mu}\|_1 < \|\bs{\xi}\|_1 $. Therefore, $\bs{\xi}\neq \boldsymbol{\mu}$, which implies that 
the total consumption at the Nash equilibrium $\bs{\xi}$ is not efficient in the sense of Pareto.

Since users have incentives to consume more resources, and this leads to a suboptimal outcome, we conclude that the electricity system model resembles the tragedy of the commons.
\end{proof}
}

\modification{
\begin{proof}[Proof of Theorem \ref{thm:ratio}]
 From Theorem \ref{thm:efficiency} we know that $\norm{\bs{\mu}}_1 \leq \norm{\bs{\xi}}_1 $. Therefore, 
 $\sfrac{\norm{\bs{\mu}}_1}{\norm{\bs{\xi}}_1} \leq 1$. Also, 
 $\mu_i \leq \xi_i$.
 implies that $\dot{v}_i(\mu_i) \geq \dot{v}_i(\xi_i)$, since $\dot{v}_i(\cdot)$ is concave and increasing. From the  FOC in Eq.~(\ref{eq:FOC1}) and Eq.~(\ref{eq:FOC2}) we deduce that
\begin{equation} \label{eq:ineq2}
 p\left( \norm{ \bs{\xi} }_1 \right) +  \xi_i \, \dot{p}  \left( \norm{ \bs{\xi} }_1 \right)
 \leq
 p\left( \norm{ \bs{\mu} }_1 \right) +  \norm{\bs{\mu}}_1 \dot{p}\left( \norm{ \bs{\mu} }_1\right).
\end{equation}

On the other hand, we can use the convexity of the price function to show that 
\begin{equation}\label{eq:price_approx}
 p\left( \norm{ \bs{\mu} }_1 \right) +   \left( \norm{\bs{\xi}}_1 - \norm{\bs{\mu}}_1 \right) \,  \dot{p}  \left( \norm{ \bs{\mu} }_1 \right) \leq p \left( \norm{ \bs{\xi} }_1 \right).
\end{equation}
We can use Eq. (\ref{eq:ineq2}) and (\ref{eq:price_approx}) to obtain
\begin{equation}\label{eq:demand_relation_a}
 \norm{\bs{\xi}}_1 \dot{p}\left( \norm{ \bs{\mu} }_1\right) + \xi_i \, \dot{p}  \left( \norm{ \bs{\xi} }_1 \right)
  \leq
 2 \norm{\bs{\mu}}_1 \dot{p}\left( \norm{ \bs{\mu} }_1\right)
 \end{equation}
Recall that $\dot{p}(\cdot)$ is increasing, and it follows that 
$ \dot{p}  \left( \norm{ \bs{\xi} }_1 \right) \geq \dot{p}  \left( \norm{ \bs{\mu} }_1 \right)  $, which can be replaced in Eq. (\ref{eq:demand_relation_a}) to obtain 
\begin{equation}\label{eq:demand_relation_b}
 \norm{\bs{\xi}}_1 + \xi_i  \leq  2 \norm{\bs{\mu}}_1.
\end{equation}
Note that this expression is true for all $i$, because from Theorem \ref{thm:efficiency} we know that $\mu_i \leq \xi_i$. 
Hence,  we can sum Eq. (\ref{eq:demand_relation_b}) for all individuals in the population to obtain
\begin{equation}
 (N+1)  \norm{\bs{\xi}}_1  \leq  2 N \norm{\bs{\mu}}_1, 
\end{equation}
which leads to the desired result.
\end{proof}
}

\modification{
\begin{proof}[Proof of Theorem \ref{thm:budget}]
This proof is made by contradiction. First, we assume that there exist
a function $f(\cdot)$ such that the mechanism is budget balanced, i.e., 
$\sum_{i=1}^N I_i(\bs{\bs{q}})= 0.$
Now, we express the incentives in matrix form. On that purpose, 
we first define 
$
[f(\bs{\bs{q}}_{-1}),\ldots,f(\bs{\bs{q}}_{-N})]^\top = F \bs{\bs{q}},
$
as a vector with the $i\th$ element equal to $f(\bs{\bs{q}}_{-i})$. 
In particular, $F = (\bs{e} \bs{\omega}^\top - \diag(\omega_1,\ldots,\omega_N))$,
 $\bs{\omega}=[\omega_1,\ldots,\omega_N]^\top$, 
$\diag(\omega_1,\ldots,\omega_N)$ is a diagonal matrix, and 
 $\bs{e}$ is a vector in $\mathbb{R}^N$ with all its components equal to 1.

Since $p(\cdot)$ is an affine function, Eq. (\ref{eq:I_i}) can be expressed as
$
\sum_{i=1}^N I_i(\bs{\bs{q}}) =  \beta \sum_{i=1}^N \Big( \sum_{j\neq i}^N q_j \Big) 
\Big( f(\bs{\bs{q}}_{-i}) -  q_i \Big) .$
This can be rewritten in matrix form as
$
\sum_{i=1}^N I_i(\bs{\bs{q}}) =\beta \bs{\bs{q}}^\top \Phi ( F \bs{\bs{q}} - \bs{\bs{q}})  ,
$
where $\Phi = (\bs{e}\bs{e}^\top - I)$ 
and $I$ is the identity matrix in $\mathbb{R}^{N\times N}$.

Now, considering the budget balance condition, 
we have
$
\bs{\bs{q}}^\top \Phi  F \bs{\bs{q}} = \bs{\bs{q}}^\top \Phi \bs{\bs{q}}.
$
This equation is satisfied if either $q_i=0$ for all $i\in \mathcal{P}$, or  
if $F = I$. Note that $F$ is a matrix with zeros in the diagonal, therefore, $F\neq I$. Accordingly, none of the previous conditions are satisfied
 for all vector $\bs{\bs{q}}\in \mathbb{R}_{\geq 0}^{N}$.
Consequently,  we conclude that
the budget balance property cannot be achieved with the incentives mechanism described by Eq. (\ref{eq:I_i}), (\ref{eq:h}), and (\ref{eq:f}).
\end{proof}
}

\begin{proof}[Proof of Proposition \ref{prop:f}]
Let us consider an arbitrary consumption profile $\hat{\bs{q}}$ in $\mathbb{R}_{\geq 0}^N$ such that $\hat{q}_i = \hat{q}_j$, for some $i, j \in \mathcal{P}$.
Since the average cost price signal is an affine function, the incentives function in Eq. (\ref{eq:I_i}) can be rewritten as 
$
I_i(\hat{\bs{q}}) = \beta \big( \sum\nolimits_{h\neq i} \hat{q}_j \big) 
\big( f(\hat{\bs{q}}_{-i}) - \hat{q}_i \big) .
$
If we use an incentives scheme with $f(\cdot)$ defined by Eq.~(\ref{eq:f_ideal}), then the incentives assigned to the $i\th$ and $j\th$ agent are
\begin{equation}\label{eq:Ii}
I_i(\hat{\bs{q}}) = \beta \left( \sum\nolimits_{h\neq i} \hat{q}_h \right) 
\left(   \frac{1}{N-1} \sum\nolimits_{h\neq i} \hat{q}_h  - \hat{q}_i \right),
\end{equation}
\begin{equation} \label{eq:Ij}
I_j(\hat{\bs{q}}) = \beta \left( \sum\nolimits_{h\neq j} \hat{q}_h \right) 
\left(   \frac{1}{N-1} \sum\nolimits_{h\neq j} \hat{q}_h  - \hat{q}_j \right).
\end{equation}
Since $\hat{q}_i=\hat{q}_j$, then $\sum_{h\neq i} \hat{q}_h = \sum_{h\neq j} \hat{q}_h$. Hence, $I_i(\hat{\bs{q}}) = I_j(\hat{\bs{q}})$ and  condition $(i)$ is satisfied.

Now, if the consumption profile $\hat{\bs{q}}$ satisfies $\hat{q}_i = \hat{q}_j = \sigma$ for all $i,j\in\mathcal{P}$, then 
$$
I_i(\hat{\bs{q}}) = \beta \left( (N-1) \sigma \right) 
\left(   \frac{N-1}{N-1} \sigma   - \sigma \right) = 0.
$$
Consequently, condition $(ii)$ is satisfied.

Finally, let us consider an arbitrary vector $\hat{\bs{q}}$ such that $\hat{q}_i > \hat{q}_j$ for some $i,j \in \mathcal{P}$. Then we know that 
\begin{equation}\label{eq:sum1}
\sum\nolimits_{h\neq i} \hat{q}_h < \sum\nolimits_{h\neq j} \hat{q}_h.
\end{equation}
Furthermore, 
$
\frac{1}{N-1} \sum_{h\neq i} \hat{q}_h + \hat{q}_j< \frac{1}{N-1}  \sum_{h\neq j} \hat{q}_h + \hat{q}_i,
$
that can be rewritten as
\begin{equation}\label{eq:sum2}
\frac{1}{N-1} \sum\nolimits_{h\neq i} \hat{q}_h - \hat{q}_i < \frac{1}{N-1}  \sum\nolimits_{h\neq j} \hat{q}_h - \hat{q}_j.
\end{equation}
Inequalities in Eq.~(\ref{eq:sum1}) and (\ref{eq:sum2}) can be used with Eq.~(\ref{eq:Ii}) and (\ref{eq:Ij}) to show that $I_i(\hat{\bs{q}}) < I_j(\hat{\bs{q}})$. Hence, property $(iii)$ is satisfied.
\end{proof}

\begin{proof} [Proof of Proposition \ref{prop:subsidies}]
First, 
consider $q_i^2 + q_j^2 - 2q_i q_j = (q_i-q_j)^2 \geq 0$ for all $q_i\in \mathbb{R}_{\geq 0}^{T}$. Hence, we have that
$
q_i^2 + q_j^2 \geq 2q_i q_j.
$
Now, summing in both sides of the previous equation we obtain
$
\sum_{i\in\mathcal{P}} \sum_{j\neq i} (q_i^2 + q_j^2) \geq \sum_{i\in\mathcal{P}} \sum_{j\neq i} 2q_i q_j,
$
which is equivalent to
$
(N-1) \sum_{i\in\mathcal{P}} q_i^2 \geq \sum_{i\in\mathcal{P}} \sum_{j\neq i} 2q_i q_j.
$
Reordering results
\begin{equation}\label{eq:ineq}
 \sum\nolimits_{i\in\mathcal{P}} q_i^2 \geq \frac{2}{N-1} \sum\nolimits_{i\in\mathcal{P}} \sum\nolimits_{j\neq i} q_i q_j.
\end{equation}
Now, 
 let $A_{j,i} = \frac{-1}{N-1}$ if $i\neq j$ and $A_{i,i} = 1$ for all 
$i,j\in\mathcal{P}$. Therefore, the incentives in Eq. (\ref{eq:I_i_ideal_matrix}) can be expressed as $\beta \bs{q}^\top A \bs{q} = \beta \sum_{i\in\mathcal{P}} q_j \Big( \sum_{j\in\mathcal{P}} q_j A_{j,i} \Big).$
This can be rewritten as
$$
\beta \bs{q}^\top A \bs{q} =\beta \Bigg( \sum\nolimits_{i\in\mathcal{P}} q_i^2 + \frac{-1}{N-1} \sum\nolimits_{i\in\mathcal{P}} q_j \Big( \sum\nolimits_{j\neq i} q_j  \Big) \Bigg).
$$
From Eq. (\ref{eq:ineq}), it can be seen that 
$
\bs{q}^\top A \bs{q} \geq 0,
$
for all $\bs{q} \in\mathbb{R}_{\geq 0}^{N}$.
\end{proof}

\modification{

\begin{proof} [Proof of Proposition \ref{prop:improvement}]
Recall from Theorem \ref{thm:efficiency}  that the aggregate surplus at the optimal outcome $\bs{\mu}$ is greater than aggregate  surplus at the Nash equilibrium $\bs{\xi}$, that is
$\sum_{i\in\mathcal{P}} U_i( \bs{\mu} ) > \sum_{i\in\mathcal{P}} U_i( \bs{\xi} )$.
 Also, recall that the system with incentives achieves the optimum outcome, then the aggregate surplus of the system with incentives is  
 $\sum_{i\in\mathcal{P}} W_i (\bs{\mu}) = \sum_{i\in\mathcal{P}} \left( U_i( \bs{\mu} ) + I_i(\bs{\mu}) \right)$.

 If we assume that the incentives are founded by the ISO, then the aggregate surplus with the mechanism is 
 %
  $$\sum\nolimits_{i\in\mathcal{P}} W_i (\bs{\mu})  - I_i(\bs{\mu})= \sum\nolimits_{i\in\mathcal{P}} U_i( \bs{\mu} ) 
  > \sum_{i\in\mathcal{P}} U_i( \bs{\xi} ).$$
 Consequently, the aggregate surplus with the mechanism is equal to the aggregate surplus of the optimal equilibrium. 
\end{proof}

}

\begin{proof}[Proof of Theorem \ref{thm:rationality}]
 Let us consider a consumption profile in which the $i\th$ individual is not consuming energy, i.e., $\bs{q}$ such that $q_i = 0$ for some $i\in\mathcal{P}$. Note that $\norm{\bs{q}}_1 = 0 + \norm{\bs{q}_{-i}}_1$. Hence, from Eq.~(\ref{eq:u_incentives}) we have that the surplus with incentives is equal to
 \begin{multline}
  W_i (q_i = 0, \bs{q}_{-i}) = v_i(0) +
  \\
  \norm{\bs{q}_{-i} }_1  \left(p\left(\frac{N}{N-1} \norm{\bs{q}_{-i} }_1 \right) - p\left(\norm{\bs{q}_{-i} }_1 \right)  \right)
 \end{multline}
Since the price function is increasing, we know that $p\left(\frac{N}{N-1} \norm{\bs{q}_{-i} }_1 \right) \geq p\left(\norm{\bs{q}_{-i} }_1 \right) $. Consequently, the surplus function of every agent $i$ is greater or equal than zero. 
 \end{proof}

\begin{proof}[Proof of Theorem \ref{thm:benefit_agents}]
 Let us rewrite the surplus of the $i\th$ agent (see Eq.~(\ref{eq:u_i_})) at the equilibrium as
 $$
 U_i(\bs{\mu}) = v_i(\mu_i) - \norm{\mu}_1 p(\norm{\mu}_1) + \norm{\mu_{-i}}_1 p(\norm{\mu}_1).
 $$
 If we evaluate the surplus with incentives (see Eq.~(\ref{eq:u_incentives})) in the Pareto optimal outcome $\bs{\mu}$, we can find that 
 \begin{multline}
 W_i(\bs{\mu}) - U_i(\bs{\mu}) = 
 \\
 \norm{\bs{\mu}_{-i}}_1 \left( p\left( \frac{1}{N-1} \norm{\bs{\mu}_{-i}}_1 + \norm{\bs{\mu}_{-i}}_1 \right) 
  - p( \norm{\bs{\mu}}_1 )   \right).
 \end{multline}


 Now, if  $\mu_i < \frac{1}{N-1} \sum_{h\neq i} \mu_h$ ( that can be rewritten as  $\mu_i < \frac{1}{N} \sum_{h\in\mathcal{P}} \mu_h$), then 
 $\frac{1}{N-1} \norm{\bs{\mu}_{-i}}_1 + \norm{\bs{\mu}_{-i}}_1 > \norm{\bs{\mu}}_1$ and consequently 
 $W_i(\bs{\mu}) - U_i(\bs{\mu}) > 0$.
 
 On the other hand, if  $\mu_i \geq \frac{1}{N-1} \sum_{h\neq i} \mu_h$ (that can be rewritten as $\mu_i \geq \frac{1}{N} \sum_{h\in\mathcal{P}} \mu_h$), then $\frac{1}{N-1} \norm{\bs{\mu}_{-i}}_1 + \norm{\bs{\mu}_{-i}}_1 \leq \norm{\bs{\mu}}_1$ and thus 
 $W_i(\bs{\mu}) - U_i(\bs{\mu}) \leq 0$.

\end{proof}

\bibliographystyle{IEEEtran} 
\bibliography{references}

\begin{IEEEbiographynophoto}
{Carlos Barreto}
received the B.S. degree in electronic engineering from Universidad Distrital Francisco Jos\'e de Caldas, Bogot\'a, Colombia in 2011. In 2013 he received his M.S. degree in electronic engineering from Universidad de los Andes, Bogot\'a, Colombia. He is currently working toward the Ph.D. degree in the department of Computer Science, University of Texas at Dallas. 

From 2013 to 2014, he was a Young  Researcher  with  the  Department  of  Electrical  and  Electronics  Engineering, Universidad  de  los  Andes.
Since 2014 he has been a 
Research Assistant at University of Texas at Dallas, Richardson, TX, USA. His research interests include cyber physical systems security, distributed resource allocation, and game theoretic methods with applications to smart grids.
\end{IEEEbiographynophoto}
\begin{IEEEbiographynophoto}
{Eduardo Mojica-Nava}
(S'00–M'10)  received  the B.S.   degree   in   electronics   engineering  from   the Universidad  Industrial  de  Santander,  Bucaramanga, Colombia, in 2002, the M.Sc. degree in electronics engineering  from  the  Universidad  de  Los  Andes, Bogota, Colombia, and the Ph.D. degree in automatique et informatique industrielle from the E\'cole des Mines de  Nantes,  Nantes,  France  in  co-tutelle  with Universidad de Los Andes, in 2010. 

He  was  an  Assistant  Professor  with  Universidad Catolica   de   Colombia,   Bogota,   Colombia.   From 2011  to  2012,  he  was  a  Post-Doctoral  Researcher  with  the  Departmentof  Electrical  and  Electronics  Engineering,  Universidad  de  los  Andes.  He is  currently  an  Associate  Professor  with  the  Department  of  Electrical  and Electronics Engineering, National University of Colombia, Bogota, Colombia. His  current  research  interests  include  optimization  and  control  of  complexnetworked systems, switched and hybrid systems, and control in smart grids applications.
\end{IEEEbiographynophoto}
\begin{IEEEbiographynophoto}
{Nicanor Quijano}
(S’02-M’07-SM’13) received his B.S. degree in Electronics Engineering from Pontificia Universidad Javeriana (PUJ), Bogotá, Colombia, in 1999. He received the M.S. and PhD degrees in Electrical and Computer Engineering from The Ohio State University, in 2002 and 2006, respectively. In 2007 he joined the Electrical and Electronics Engineering Department, Universidad de los Andes (UAndes), Bogotá, Colombia as an Assistant Professor. In 2008 he obtained the Distinguished Lecturer Award from the School of Engineering, UAndes. He is currently an Associate Professor, the director of the research group in control and automation systems (GIAP, UAndes), and a member of the Board of Governors of the IEEE CSS for the 2014 period. He was the chair of the IEEE Control Systems Society (CSS), Colombia for the 2011-2013 period.  His research interests include hierarchical and distributed optimization methods, using bio-inspired and game-theoretical techniques for dynamic resource allocation, applied to problems in energy, water, and transportation. For more information and a complete list of publications see: \url{http://wwwprof.uniandes.edu.co/~nquijano}.
\end{IEEEbiographynophoto}
\end{document}